\documentclass[12pt]{article}
\usepackage{authblk}
\usepackage{amssymb}
\usepackage{amsmath}
\usepackage{mathrsfs}
\usepackage{amsfonts}
\usepackage{txfonts}
\usepackage{braket}
\usepackage{color}
\usepackage{graphicx}
\usepackage{tikz,pgffor}
\usepackage{rotating}
\usetikzlibrary{calc, shapes, matrix,snakes,positioning,backgrounds,arrows,automata,backgrounds,fit,decorations.pathreplacing}
\usetikzlibrary{shadows}
\usepackage{smartdiagram}
\usepackage{tkz-graph}
\GraphInit[vstyle = Shade]
\usepackage[latin1]{inputenc}
\usepackage{pgf}
\newtheorem{thm}{Theorem}

\newtheorem{lem}[thm]{Lemma}
\newtheorem{prop}[thm]{Proposition}
\newtheorem{exmp}[thm]{Example}
\newtheorem{defn}[thm]{Definition}
\newcommand{\norm}[1]{\left\Vert#1\right\Vert}
\newcommand{\abs}[1]{\left\vert#1\right\vert}

\def\XXint#1#2#3{{\setbox0=\hbox{$#1{#2#3}{\int}$ }
\vcenter{\hbox{$#2#3$ }}\kern-.6125\wd0}}

\newcounter{lastnote}

\title{Formal Verification using Second-Quantized Horn Clauses}
\author{Radhakrishnan Balu}
\affil{Army Research Laboratory Adelphi, MD, 21005-5069, USA \\
       radhakrishnan.balu.civ@mail.mil}            
\date{Received: date / Accepted: \today}
\begin{document}
\maketitle

\begin{abstract}
In our previous work \cite {RB2015} we described quantized computation using Horn clauses and based the semantics, dubbed as entanglement semantics as a generalization of denotational and distribution semantics, and founded it on quantum probability by exploiting the key insight classical random variables have quantum decompositions. Towards this end we built a Hilbert space of H-interpretations and a corresponding non commutative von Neumann algebra of bounded linear operators \cite {RB2016}. In this work we extend the formalism using second-quantized Horn clauses that describe processes such as Heisenberg evolutions in optical circuits, quantum walks, and quantum filters in a formally verifiable way. We base our system on a measure theoretic approach to handle infinite dimensional systems and demonstrate the expressive power of the formalism by casting an algebra used to describe interconnected quantum systems (QNET) \cite {Gough2008} in this language. The variables of a Horn clause bounded by universal or existential quantifiers can be used to describe parameters of optical components such as beam splitter scattering paths, cavity detuning from resonance, strength of a laser beam, or input and output ports of these components. Prominent clauses in this non commutative framework are Weyl predicates, that are  operators on a Boson Fock space in the language of quantum stochastic calculus, martingales and conjugate Brownian motions compactly representing statistics of quantum field fluctuations. We formulate theorem proving as a quantum stochastic process in Heisenberg picture of quantum mechanics, a sequence of goals to be proved, using backward chaining.
\end{abstract}
\section{Introduction}
Categorification of quantum mechanics in \cite {Abramsky2004} and \cite {Abramsky2008} paved the way for expressing quantum computation and protocols amenable to visual verification. Rigorous certification of quantum information is an active area of research as in formal verification of quantum processes \cite {Deng2012}, topological quantum computation \cite {vicary2015}, quantum lambda calculus based approaches \cite {Sellinger2006}, and quantum prolog where predicates are quantized using the annealing paradigm \cite {Scott2018}. Formal methods also play a critical role in decidability of algorithms, either classical or quantum, as exemplified by the tools of partial recursive functions \cite {Svozil2010}. Recently, Schulman and Schreiber integrated mathematical logic and quantum field theory \cite {Urs2012} using cohesive Homotopy type theory. Here we develop predicates for infinite dimensional quantum systems to express dynamical aspects of quantum fields such as creation, annihilation, and counting processes that can be related to theorem proving in mathematical logic using the language of higher categories \cite {Beiz}. The relationship between these seemingly disparate notions is established via dagger compact symmetric monoidal categories. Intuitionistic or linear logic can be used to construct a symmetric monoidal category \cite {Girard} and by basing it on a Hilbert space we can build a compact symmetric monoidal category.  By realizing this system on a symmetric Fock space we can construct a dagger compact symmetric monoidal category. 

Our framework is based on Horn clauses and the theorem proving is a backward chaining similar to prolog we have to use linear logic to treat proofs and refutations in equal footing which is also consistent with the reversible evolution of quantum processes. A proof of the  Horn clause, we use "lollipop" for implication operator, $R\multimapinv Q \boxtimes P$ \cite{Shulman2018} consists of either one of converting the proofs of P and Q into that of R,  proof of P and refutation of R into refutation of Q, proof of P and refutation of Q into refutation of R. The refutation of the Horn clause consists of proofs of P and Q and refutation of R. An equivalent Horn clause is $(R\multimapinv Q )\multimapinv P$ as we use multiplicative conjunction for the antecedents. For the generic Horn clause $R\multimapinv (P_1 \boxtimes P_2 \boxtimes \dots \boxtimes P_n)$ a proof consists of all possible direct or contrapositive proofs that contradict one of the antecedents. Contradicting a predicate P is same as developing a proof for P that leads to absurdity. The theorem prover will start with the negation of the goal to be established and try to prove the antecedents until an absurdity is generated.  

In graphical calculus based on higher categories integers encoding classical states form objects, states of Hilbert spaces are 1-morphisms (wires), and unitary operators at the vertices cascade evolution as 2-morphisms.  In our framework, the members of Herbrand universe form the objects, quantum predicates are the morphisms, and the stochastic processes form morphisms on morphisms consistent with the hierarchies of the 2-Hilb category. That is, the morphisms are not invertible resulting in a groupoid rather they are conjugates, adjoints, and second-quantized adjoints as in higher Hilb categories. In this spirit we will categorify standard quantum logics by introducing the category StdLog and enrich it with structures. We build a 2-Chu space out of it where the refutations are quantum stochastic processes in Heisenberg picture and the proofs are evolutions in Schr$\ddot{o}dinger$ picture. We could iterate Chu space construction, on unitaries that encode algorithms, forming a hierarchy of spaces of arbitrary hight with duality built in. The Fock space can also be used to model open systems and processes constructed out of this space are equivalent to manipulation of morphisms in a graphical calculus. It is easy to see the correspondence as a Fock space is built by countable direct sum of finite tensor product of Hilbert spaces. Accordingly, we will decorate our predicates as $\overset{0}{p}$ for classical predicates (deterministic or randomized) and use the decoration $\overset{1}{p}$  for quantum operators, $\overset{2}{p}$  for unitary and anti-unitary operators that evolve quantum states, and $\overset{3}{p}$  for second quantized predicates, these are the Weyl predicates 2-morphisms cascaded vertically along the time axis. The intent is to view n-decorations as morphisms between spaces where (n-1)-decorations inhabit stratifying the space where proofs flow. These probability predicates represent statistical ensembles form spaces, not necessarily smooth, that are models of the logic. In such a view proofs and refutations are continuous paths on these synthetic topological spaces. In our framework antecedents of Horn clauses can support braiding as $(\overset{1}{p} \multimapinv \overset{1}{p}\otimes\overset{1}{r}) = (\overset{1}{p} \multimapinv \overset{1}{r}\otimes\overset{1}{p})$ and the compactness arise as $(\overset{1}{p} \multimapinv \overset{1}{q}) = (\overset{1}{p} \otimes \overset{1}{q}^\dag \multimapinv)$. Our focus in this work is the logical language aspects to express solutions based on second quantization and we provide several well known examples in this formalism, incorporating compact symmetric monoidal category informally and postponing a formal semantics based on 2-Category theory \cite {Vicary2008} and relativistic extensions \cite{BaluRel2018}, systems of imprimitivity based on induced representations of groups generalized to categories, to future work along with the description of probabilistic predicates as 3-morphism. An outline of such a program would start with the identification of the Euclidean group of $\mathscr{H}\times\mathscr{U}(\mathscr{H})$ for an arbitrary Hilbert space $\mathscr{H}$ as a connected Lie group \cite {KP1992}. Members of this group can be used to build Weyl unitaries of the quantum stochastic calculus. Induced representations from the subgroup of this Euclidean group can be expressed in the language of category theory using the underlying fiber bundle structure. With our identification of Weyl predicates to the corresponding second quantized operators we can build categories for quantum Horn clauses describing fields. We have taken liberties with the notation at some places dropping the decorations,  using comma, $\boxtimes$, and $\otimes$ interchangeably and $\multimapinv$ for $\multimapinv$ when the context is revealing.

We present our formalism in the context of quantum logic by identifying the non Boolean lattice underlying the language to construct quantum stochastic processes as proofs. In this instance violation of distributive property results in a non Boolean logic. It is important to note that there are difficulties in defining the tensor products of lattices and logics \cite{RF81}. Consequently, we shall consider only the standard logics, that is the lattices and logics defined on separable Hilbert spaces, in our discussion. We resort to quantum probability defined on separable Hilbert spaces in describing composite systems to avoid the issues related to tensor products. The separability is required for the Gleason{'}s theorem to hold which is at the heart of quantum probability. The lattice operations, join and meet of the well formed formula (WFF), are defined in terms of set union and intersections of H-interpretations and inclusion in sets, rigorous definitions are provided later, would correspond to logical implication enabling the definition of Horn clauses. In quantum probabilistic logic  stochastic realizations are absent as the grounding of predicates occur only at the time of quantum measurements. We find that quantum probability (QP) and first order quantum logic (FOQL) are related formalisms to express high level computation, with QP suitable for describing time evolutions, and FOQL provides language constructs for computations as a sequence of operators acting upon a quantum state. Quantum probability space is a composition of several classical probability spaces that could support both continuous and discrete measures simultaneously. Quantum mechanical measurements will result in a single classical space consisting of compatible, that is simultaneously observable, operators. In the quantum logic perspective the system will be described as a non Boolean lattice  and measurements will resolve it into a single logic among several possible classical Boolean logics. We refer the readers to \cite{KP1992} and \cite{PA1995} for an excellent introduction to quantum probability and more advanced treatment of the subject respectively.The interaction between classical and quantum systems can be expressed in this framework either at the lattices level or as an embedded process within a quantum process. This formalism may not be as intuitive as the ones based on pictorial reasoning but can deal with infinite-steps processes as our framework seamlessly integrates with measure spaces.

The deductive process in this system of logic can be visualized as a sequence of predicates  to be proved starting from an initial one and identifying a target predicate. We need a symmetric Fock space to form a sequence of observables corresponding to the chain of predicates forming the deduction. Exploiting the true randomness of microscopic systems and the fact that any classical random variable can be decomposed into quantum operators \cite{OB} we build observables on a Fock space to represent predicates. We express a quantum walk in this framework and trace the sequence of quantum predicates to be proved, as part of theorem proving, as a discrete-time quantum stochastic process. With experimental realization of quantum walks on different platforms our approach will lay the foundation for at least a class of Markovian theorem provers. 

\begin{exmp}
Transfer of quantum states via teleportation is an important step in many communication protocols and heralded entanglement \cite{MUNROE} is one of the well known implementations. Let us consider a slightly more complex example of Heralded entanglement of matter qubits or solid state (SS) systems $sQ_A$ and $sQ_B$ separated in space and correlated by photonic qubits $pQ_A$ and $pQ_B$ (flying qubits). Qubits entangled at a distance are an important ingredient in achieving distributed computation \cite{HERALD}.The steps involved are to prepare a pair of superposition of spin states (SS) separated by a distance, apply Hadamard gate to entangle spin and photon at each site, coincidence measurement via a beam splitter (bs) of the two flying qubits to project on to SS qubits that will finally entangle the remote SS qubits. Designing predicates and Horn clauses is an art as it encodes algorithmic aspects in a logical programming language with very few constructs that becomes subtler in the quantum context. Here, the core steps of the protocol are given, the asterisks indicates a measurement step,  and after a formal introduction to quantum probability let us discuss another protocol that involves conditional quantum states. 
\begin {align*}
 \overset{2}{e}ntangle(sQ, pQ) & \multimapinv\text{  } \overset{2}{h}adamard(sQ, pQ)\boxtimes (state(sQ) \overset{1}{=} \psi), (state(pQ) \overset{1}{=} \phi).\\
\overset{2}{c}oincide(sQ_A,sQ_B,pQ_A,pQ_B) & \multimapinv\text{   }\overset{2}{e}ntangle(sQ_A,pQ_A)\boxtimes\text{    }\overset{2}{e}ntangle(sQ_B,pQ_B)\boxtimes\\
                    & \overset{1}{b}s(sQ_A,sQ_B,pQ_A,pQ_B)\boxtimes(state(sQ) \overset{1}{=} \psi_A)\boxtimes (state(pQ) \overset{1}{=} \phi_A)\boxtimes \\
                    & (state(sQ) \overset{1}{=} \psi_B)\boxtimes (state(pQ) \overset{1}{=} \phi_B).\\
\overset{2}{h}erald(sQ_A,sQ_B,pQ_A,pQ_B) & \multimapinv \text{  }\overset{2}{c}oincide(sQ_A,sQ_B,pQ_A,pQ_B)\boxtimes\\
                   & (state(\Psi) \overset{1}{=} 1/\sqrt{2}(\ket{1_{A}0_B}+\ket{0_{A}1_B}))\boxtimes \overset{2}{p}roject(sQ_A,sQ_B,pQ_A,pQ_B,\Psi)^*.
\end {align*}
Initial quantum state:  $1/\sqrt{2}(\ket{\uparrow}+\ket{\downarrow})\otimes{1/\sqrt{2}}(\ket{\uparrow}+\ket{\downarrow})$. So, we can start the theorem proving with the predicate on two Nitrogen Valence centers as $herald(nv_1,nv_2)$, $state(nv_1) = 1/\sqrt{2}(\ket{\uparrow}+\ket{\downarrow})$, $state(nv_2) = {1/\sqrt{2}}(\ket{\uparrow}+\ket{\downarrow})$) that will result in the solid state qubits entangled.\\
\end {exmp}

\section {Formalism: StdLog Category}

We have reviewed the usual notions of quantum logic in the appendix and here we provide the categorification of the same. The objects of the category StdLog are standard logics, that is orthomodular lattices corresponding to separable Hilbert spaces. It is essentially the category of Hilbert spaces Hilb but with an emphasis on the underlying lattice structure. Linear maps between them form the morphisms. The zero object is the lattice corresponding to the zero dimensional Hilbert space, the coproduct addition, the cokernel subtraction are all similar to the Hilb category. StdLog is also compatible with the multiplicative structure and the categorification of the inner product is the hom functor. Analogously, in the 2-StdLog category the dual morphisms (adjoint operators) are present via the adjoint functors $F^* :D\rightarrow C$ and $F:C\rightarrow D$ when there is a natural isomorphism $hom(Fc, d) \cong hom(c, F^* d)$ 

\begin{defn} A StdLog-category is a category C such that any pair of objects $x,y\in C$ the set of mophisms hom(x, y) is equipped with the structure of a standard logic, and for any objects $x, y, z \in C$ the composition map 
$o: hom(x, y) \times hom(y, z) \rightarrow hom (x, z)$ is bilinear.
\end {defn}
This category is endowed with a covariant functor $\bar{  }: StdLog \rightarrow StdLog$
\begin{defn} A $S^*$-category is a StdLog-category with an involution that defines an antinatural transformation from hom(x, y) to $\overline{hom(y, x)}$. 
\end {defn}
\begin{defn} Let x and y be objetcs of a *-category. A morphism $u:x \rightarrow y$ is unitary if $uu^* = 1_x$ and $u^* u = 1_y$. A morphism $a:x \rightarrow x$ is self-adjoint if $a^* = a$.
\end {defn}
StdLog is a symmetric monoidal category with $\mathbb{C}$ as the unit object with the tensor product of the separable Hilbert space. Let us make this a *-category by defining an antinatural transformation fro hom(x, y) to $\overline{hom(y, x)}$ and it can be shown to be antiunitary. In this *-category every isomorphism is defined by a unitary.
\begin {defn} A quantum observable Obs is a functor from the lattice of  Borel sets $\mathscr{B}(\mathbb{R})$ of the real line $\mathbb{R}$ into StgLog.
$Obs:\mathscr{B}(\mathbb{R}) \rightarrow StdLog$ such that the morphisms State(x) with $x\in Obs$ is a probability measure and in additin the morphism x preserves the lattice structure.
\end {defn} 
\begin {defn} Given the set of all observables and a measure a quantum state State is an adjoint functor to Obs which is a morphism from the lattice of  StgLog into Borel sets $\mathscr{B}(\mathbb{R})$ of the real line $\mathbb{R}$.
$State:StdLog \rightarrow \mathscr{B}(\mathbb{R})$ such that the morphisms State(o) with $o\in Obs$ are probability measures that map elements of StdLog into that of the Borel lattice. Let us denote this adjunction by the symbol $\iota$
\end {defn} 
Let us build a Chu construction $Chu(StdLog, \iota)$ as follows:
\begin {lem} The triples $(\text{Obs, State, and the relation }\iota)$ form the Chu space $Chu(StdLog, \iota)$ for the standard logics.
\end {lem}
Now we can use Piron's theorem to construct a 2-Chu space and so let us recollect the result.
\begin {thm} Piron \cite {Piron1964} Let $\mathscr{L}$ be any logic. Then, a necessary and sufficient condition that $\mathscr{L}$ is isomorphic to the logic of all closed linear manifolds of a separable Hilbert space is that $\mathscr{L}$ be a projective logic and have the property that every family of mutually orthogonal points of $\mathscr{L}$ be at most countable.
\end {thm}
\begin {lem}
3-PrjSpc is a 3-category that is equal to 2-StdLog up to an isomorphism. $Chu(PrjSpc, \bot)$ is a 2-Chu space.
\end {lem}
Let us state the salient features of the formalism here:

The logical formulae are expressed as Horn clauses with one atomic predicate for the consequent and multiple predicates connected by conjunction for antecedent: $q\multimapinv q_1\otimes q_2 \otimes \dots \otimes p_n$.

Predicates are decorated as $\overset{0}{p}$ for classical predicates (deterministic or randomized) and the decoration $\overset{1}{p}$  for quantum operators, $\overset{2}{p}$  for unitary operators that evolve quantum states, and $\overset{3}{p}$  for second quantized predicates (Weyl predicates). 

 Antecedents of Horn clauses can support braiding as $(\overset{1}{p} \multimapinv \overset{1}{p}\otimes\overset{1}{r}) = (\overset{1}{p} \multimapinv \overset{1}{r}\otimes\overset{1}{p})$ and the compactness of the category is due to the identity $(\overset{1}{p} \multimapinv \overset{1}{q}) = (\overset{1}{p} \otimes \overset{1}{q}^\dag \multimapinv)$.

Though $\overset{1}{p}, \overset{2}{p},$ and $\overset{3}{p}$ are morphisms our casting them as predicates can be viewed as "internal homs" of the category.

A proof is a martingale based on the observed predicates and time reversibility is supported by backward and forward chaining of the proofs.

Our framework is constructive in nature that doesn{'}t support law of excluded middle and therefore any resulting proof or model can be interpreted as computation.

Implications $\multimapinv$ are the only 2-morphisms as others are internal homs so that implication of proofs and refutations are compactly represented as in a constructive logic.

Our proofs respect the duality of $Schr\ddot{o}dinger$ and Heisenberg evolutions and can be cast in either pictures owing to the Chu construction.

We treat proofs and refutations with equal footing which is consistent with the time reversibility of quantum processes and the connectives of linear logic are used to accomplish this requirement.
\

\begin {exmp} Measurement of incompatible observables through a probe.

This example from \cite{LUC2006} quantum filtering work describes a way to measure interfering observables by copying it to a probe observable on a different Hilbert space. This method doesn't violate the no cloning theorem as the probe reproduces just the statistics of the original observable on a different Hilbert space enabling the simultaneous measurement of the two.

The system, probe, and the composite system can be described by the quantum probability space \\
$(\mathscr{N}=\mathscr{C}^n\otimes{\mathscr{N}_p}=\mathscr{C}^m,\mathbb{P}(X)=Tr\{\rho{X}\}\otimes{\mathbb{P}_p})$ with p as a pre determined eigen value and \\
$A=\sum_{a\in{spec(A)}}$ is an observable with m eigen values.We want to copy A into A{'} so that the unitary evolution results in same statistics for the two such as $(A\otimes{I})=(U*(I\otimes{A{'}})U)$.\\

Let us define a projection $P{'}_a = \Psi_a\Psi^*_a$ for each eigen value a of the observable A to be probed.  The unitary evolution that will give us the same statistics for both the observables is
$U=\sum_{a\in{spec(A)}}P_a\otimes{X{'}_{ap}}$ where $X{'}_{ab}$ is given by:
\begin {equation}
X{'}_{ab} = \Psi_b\Psi^*_a+\Psi_b\Psi^*_a+\sum_{c\neq{a,b}}\Psi_c\Psi^*_c;X{'}_{aa}=I \\
\end {equation}
To see that the unitary interaction copies the observable A to probe A{'}, we see that
\begin {equation}
U^*(I\otimes{P{'}_c})U = P_c\otimes{P{'}_p}+ (1 - P_c)\otimes{P{'}_c} \text{  if  } (c \neq{p}) \\
\end {equation}
and when c=p $U^*(I\otimes{P{'}_c})U =\sum_a{P_a\otimes{P{'}_a}}$ and it is easy to verify that U is an unitary operator. Finally, we have
\begin {equation}
\frac{(\mathscr{P}\otimes{\mathscr{P}_p})(U^*(I\otimes{P{'}_c})U)(P_c\otimes{I})}{(\mathscr{P}\otimes{\mathscr{P}_p})(P_c\otimes{I})} = 1, \forall{c}\\
\end {equation} 
To see that, when $c\neq{p}$ the numerators in the state $\mathscr{P}_c\otimes{\mathscr{P}{'}_p}$ is $(P_c\otimes{P{'}_p})(P_c\otimes{I})$ as the other terms are zero. This gives the ratio of 1 and for the case $c = {p}$  the numerator has $\sum_a{P_a}\otimes{P{'}_p} = I\otimes{P{'}_p}$ which evaluates to 1 in this state. Therefore, the ratio is 1 $\forall{c}$.

What we have shown is that the conditional probability that the unitary evolution gives an outcome c given that we have observed the same value for $A\otimes{I}$ is 1 and so the unitary
has faithfully copied the observable A into A{'}.\\

The corresponding Horn clauses can be set up as, here we have decorated only the unitaries to simplify notations:
\begin {align*}
\overset{2}{U}^*(I\otimes{P{'}_c})\overset{2}{U} &= P_c\otimes{P{'}_p}+(1 - P_c)\otimes{P{'}_c}.\\
\text{probe(B,A) } &\multimapinv\text{ }[\overset{2}{U}^*(I\otimes{A'})\overset{2}{U}, \overset{2}{U}^*(B\otimes{I})\overset{2}{U}] = 0.\\
P[\overset{2}{U}^*(B\otimes{I})\overset{2}{U}|A]&\multimapinv{probe(B,A{'})},\sum_i\frac{P_p(BA_i)}{P_p(A_i)}\overset{2}{U}^*(I\otimes{A_i})\overset{2}{U},\text{spec(A) }=\{A_i\}.\\
\end {align*}
with the initial quantum state: $(P\otimes{P_p};P_p(X)=tr\{XA_{p{'}}\};p{'}=spec(A{'})).$\\

Now let us extend this to three variables probe.
\begin {align*}
(A\otimes{I}\otimes{I}) & \Leftrightarrow(\overset{2}{U}^*(I\otimes{A{'}}\otimes{I})\overset{2}{U}).\\
(B\otimes{I}\otimes{I}) & \Leftrightarrow(\overset{2}{U}^*(I\otimes{I}\otimes{B{'}})\overset{2}{U}).\\
\text{probe(X,Y{'}) } & \multimapinv\text{ }[\overset{2}{U}^*(I\otimes{Y{'}}\otimes{I})\overset{2}{U}, \overset{2}{U}^*(X\otimes{I}\otimes{I})\overset{2}{U}] = 0.\\
\end {align*}
When we instantiate X=B, Y=A and X=C, Y=B; we get the following Horn clauses: 
\begin {align*}
\text{probe(B,A) } & \multimapinv\text{ }[\overset{2}{U}^*(I\otimes{A{'}}\otimes{I})\overset{2}{U}, \overset{2}{U}^*(B\otimes{I}\otimes{I})\overset{2}{U}] = 0.\\
\text{probe(C,B) } & \multimapinv\text{ }[\overset{2}{U}^*(I\otimes{I}\otimes{B{'}})\overset{2}{U}, \overset{2}{U}^*(C\otimes{I}\otimes{I})\overset{2}{U}] = 0.
\end {align*}
with the quantum state: $(P\otimes{P_p}\otimes{P_q})$.\\

Now, if we probe the observable A with B and then measure A we will get the same answer as measuring A before the probe. However, if we probe A and B in succession and measure A with C=A as a probe for B then we will get different answer. This is easy to see as
probe C disturbs A and so subsequent measurement will give a different answer for A. We can establish this very formally using the following Horn clauses:
\begin {align*}
\text{probe(B,A) } & \multimapinv\text{ }[\overset{2}{U}^*(I\otimes{A{'}}\otimes{I})\overset{2}{U}, \overset{2}{U}^*(B\otimes{I}\otimes{I})\overset{2}{U}]= 0.\\
\text{probe(A,B) } & \multimapinv\text{ }[\overset{2}{U}^*(I\otimes{I}\otimes{B{'}})\overset{2}{U}, \overset{2}{U}^*(A\otimes{I}\otimes{I})\overset{2}{U}]= 0.
\end {align*} 

When the theorem prover tries to satisfy the goal probe(C,B) proving the antecedents of the Horn clause would involve the unitary evolution  $ U^*(A\otimes{I}\otimes{I})U$ and thus ends up in disturbing the observable A.
\end {exmp}
 
 \section {Quantum walks and Theorem provers}
Quantum walks are a flexible framework for deriving algorithms relevant in quantum information processing \cite {Rad2018}, \cite {Rad2018a}, \cite {Rad2017a}, \cite {Rad2017b}and we refer the readers to the review by Salvador \cite{SV} for a comprehensive background on this topic. Walks are defined using maps that are *-homomorphisms and in the case of time-discrete evolutions they are defined using a coin operator Hilbert space in addition to the walker space. In this perspective a quantum walk is an open systems with the coin degrees of freedom acting as the bath for the walker. We will work within the open quantum formalism by realizing the effect of the coin through the operators driving the discrete quantum noises (the three martingales). We recommend the works on discrete processes \cite{KP1992, RBQA2016} for the readers to gain understanding in rigorous establishment of their existence and their applications respectively. Some essential features are provided in the Appendix for a quick reference.

Let us start with the definition of a *-homomorphism $\Theta$ required to induce the open system (quantum walk) that when divided into structure maps that have the interpretation of Lindbladian operators in the time-continuous limit of these processes. We identify the coin Hilbert space $\mathscr{H}_c$ as the d-dimensional complex space $\mathscr{H}_w = \mathscr{C}^d$ with the canonical basis \{$e_0,e_1,...,e_{d-1}$\} and the space of the walker as the set of Natural numbers $\mathscr{Z}$. 
\begin {align*}
\Theta:\mathbb{B}_0 &\rightarrow{\mathbb{B}(\mathscr{H}_w)}\otimes{\mathbb{B}(\mathscr{H}_c)}\\
\Theta &= U_c^\dagger\Lambda_w{U_c}
\end{align*}
In the above $\Lambda_w$ is a d-dimensional matrix containing information about the branching along different subspaces (or dimensions) based on the outcome of the coin. For example, if the walker takes a step in the positive direction for the coin outcome $e_0$ and the opposite direction for the other outcome then the matrix $\Lambda_w$ wil have the dimension 2x2. The structure maps \begin {align*}
\theta^j_i(X) = \mathscr{E}_{\ket{i}\bra{j}}[\Theta(X)]\text{, where  } X\in{\mathbb{B}(\mathscr{H}_w)}
\end {align*}
are projections on different subspaces of the coin conditioned upon the outcome. The map $\theta^0_0$ corresponds to coherent evolution, the maps $\theta^j_j$ drive the number martingale which in the continuos limit would become the gauge process and the rest of the form $\theta_0^j$ control the other two martingales. In the case of Hadamard walk the structure maps are
\begin {align*}
\mathbb{H}_c &=  \frac{1}{\sqrt{2}}\left( \begin{array}{cc}
1 & 1 \\
1 & -1  \end{array} \right).\\
\Lambda_w &=\ket{0}\bra{0}L^{+}+ \ket{1}\bra{1}L^{-},L^{\pm}\ket{\psi(x,n)}=\ket{\psi(\pm{x},n)}.\\
\Theta &= \Lambda_{w}o(\mathbb{I}\otimes{H_c}), \text{   o is the composition of operators}.\\
\theta^0_0 &= L^+ \text{, and  }\\
\theta^1_1 &= L^-.
\end {align*}

With this set up the existence of a discrete-time quantum stochastic process, a flow in Heisenberg picture, that is Markovian can be established \cite{KP1992}.
 There exists an operator valued process $J_n:B_0\rightarrow{B_{n]}}$ satisfying 
\begin{align*}
j_0(X) &= X\otimes{1_{[1}}.\\
j_1(X) &=\Theta(X)\otimes{1_{[2}}.\\
j_n(X) &= \sum_{0\leq{i}\leq{d-1}}j_{n-1}(\theta_i^i(X))1_{n-1]}\otimes{\ket{e_i}\bra{e_i}}\otimes{1_{[n+1}}.\\
a_n &= 1_{n-1]}\otimes{\ket{e_0}\bra{e_1}}\otimes{1_{[n+1}}.\\
a^{\dag}_n &= 1_{n-1]}\otimes{\ket{e_1}\bra{e_0}}\otimes{1_{[n+1}}.\\
j_n(X) &= j_{n-1}(Xo\Psi_R + Xo\Psi_L ))+j_{n-1}(Xo\Psi_R - Xo\Psi_L)(a_n^{\dag}a_n)\\
A_n &= a_1+...+a_n \\
A^{\dag}_n &= a^{\dag}_1+...+a^{\dag}_n \\
\Lambda_n &= a^\dagger_1+...+a^\dagger_n\\
j_n(X) &= L^+{\Psi_R}_{n-1} + L^-{\Psi_L}_{n-1}+L^+{\Psi_R}_{n-1} - L^-{\Psi_L}_{n-1}\\
\end{align*}

in the state $\rho_0\otimes{\ket{e_0}\bra{e_0}}\otimes{\ket{e_0}\bra{e_0}}\otimes\dots$, where the three noise martingales, $A_n, A^{\dag}_n, \text{ and } \Lambda_n$ evolve can be described in position space resulting in an unitary evolution below:\\
\begin {align*}
\psi_R(x,n+1) &= \psi_R(x-1,n) + \psi_L(x-1,n)\\
\psi_L(x,n+1) &= \psi_R(x+1,n) - \psi_L(x+1,n)
\end {align*}
Now, we can set up the Horn clauses to describe the Hadamard quantum walk as follows:
\begin {align*}
\overset{3}{Q}walkR(x, n+1) &\multimapinv \overset{3}{Q}walkR(x-1, n), \overset{3}{Q}walkL(x-1, n),\overset{3}{A}_n(\Omega),\overset{3}{\Lambda}_n(\Omega). \\ 
\overset{3}{Q}walkL(x, n+1) & \multimapinv \overset{3}{Q}walkR(x-1, n),  \overset{3}{Q}walkL(x-1, n),\overset{3}{A}_n(\Omega),\overset{3}{\Lambda}_n(\Omega).
\end {align*}

We described the quantum walk on a Hilbert space that is a tensor product of countable sequence of identical spaces one associates with each time instant of the evolution. We have to ensure that this space, $\mathscr{H}=\otimes_{n=0}^{\infty}\mathscr{H}_n$ embeds a quantum logic that is standard. The emphasis on standard logic is due to the importance of automorphism on logics in physics \cite {HWeyl} and not the underlying symmetries. First, we observe that infinite tensor product may lead to a Hilbert space that is not separable. We outline the essential arguments provided in Varadarajan`s book \cite {VAR} to establish that a standard logic can be carved out of symmetric Fock space and refer the reader to the book for complete details.

\begin {thm} Evolution of process $j_n$ is an automorphism on a standard logic. \\
\begin {proof}
For symmetric Fock space we require that the self adjoint operators representing observables commute with actions of permutation group $S_N$ of order N when acted upon $\mathscr{H}^N=\mathscr{H}_1\otimes\dots\mathscr{H}_N$. This space can be written as a direct sum of mutually orthogonal subspaces using the following projections:
\begin {align}
P_{\tau} &= (N!\sqrt{dim(\tau)})^{-1}  \sum_{s\in{S_N}}\tau(s)s, \text{ where } \tau \text{is any irreducible character of } S_N\\
\mathscr{H}^N(\tau) &= P_{\tau} \mathscr{H}^N \\
\mathscr{H}^N &= \oplus_\tau \mathscr{H}^N(\tau)
\end {align}
The logic of individual $\mathscr{H}^N(\tau)$ is a standard one. When the above decomposition is countably infinite \cite {VAR} we can conclude that the logic of $\mathscr{H}$ is a standard one.
Let $e_0$ and $e_1$ denote the canonical basis of the Hilbert space $\mathscr{C}^2$. An orthonormal basis of $(\mathscr{C}^2)^{\otimes{n}}$ is given by the vectors $e_{U}$ where $U\subset{\{1,2,\dots,n\}}$ and $e_U=e_{i_1}\otimes\dots\otimes{e_{i_n}}$.
In this basis the operators $A_n,A^{\dag},\Lambda_n$ at as \\
\begin {align*}
A_n{e_U} &= \sum_{k\ni{U}}e_U\cup{\{k\}} \\
A^{\dag}_n{e_U} &= \sum_{k\ni{U}}e_U\setminus{\{k\}} \\
\Lambda_n{e_U} &= (n - 2|U|)e_U
\end {align*}
We see that $A_n,A^{\dag}_n,\Lambda_n$ are all operators on a Fock space the evolution of $j_n$ can be described by an automorphism on a standard logic.  $\blacksquare$
\end {proof}
\end {thm}

 \tikzset{
  treenode/.style = {shape=rectangle, rounded corners,
                     draw, anchor=center,
                     text width=5em, align=center,
                     top color=white, bottom color=blue!20,
                     inner sep=1ex},
  decision/.style = {treenode, diamond, inner sep=0pt},
  root/.style     = {treenode, font=\Large,
                     bottom color=red!30},
  env/.style      = {treenode, font=\ttfamily\normalsize},
  JC/.style   = {root, circle, bottom color=yellow!90},
  bath/.style   = {root, bottom color=green!40},
  dummy/.style    = {circle,draw}
}
\newcommand{\yes}{edge node [above] {L}}
\newcommand{\no}{edge  node [left]  {M}}

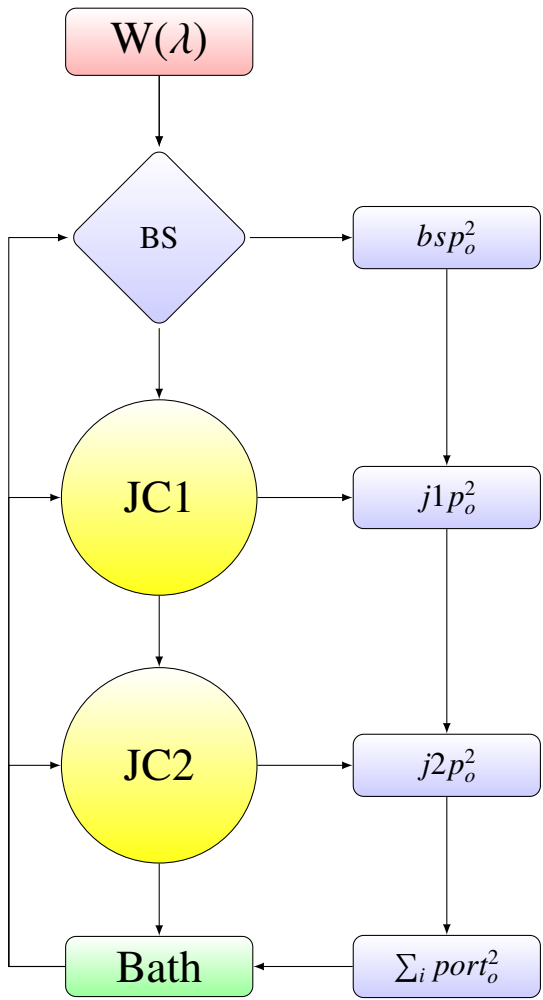
\begin{figure}
\begin {turn}{0}
\begin{tikzpicture} [-latex]

\matrix (chart)
[
      matrix of nodes,
      column sep      = 3em,
      row sep         = 5ex,
      column 1/.style = {nodes={decision}},
      column 2/.style = {nodes={env}}
    ]
    {
      |[root]| W($\lambda$)           &           \\
                   BS                           &     $bsp_o^2$     \\
       |[JC]|   JC1                          &      $j1p_o^2$     \\
       |[JC]|  JC2                           &      $j2p_o^2$     \\
       |[bath]| Bath                        &      $\sum_i port_o^2$    \\
                     };
    \draw
    (chart-1-1) edge (chart-2-1)
    \foreach \x/\y in {2/3, 3/4, 4/5} {
      (chart-\x-1) edge   (chart-\y-1) }
    \foreach \x in {2,...,4} {
       (chart-\x-1) edge (chart-\x-2) }
    (chart-5-2) edge (chart-5-1)
   ;  
     \draw
   (chart-5-1) -- +(-2,0) |- (chart-2-1)
     node[near start,sloped,above] {};
     
      \draw
   (chart-5-1) -- +(-2,0) |- (chart-3-1)
     node[near start,sloped,above] {};
   
       \draw
   (chart-5-1) -- +(-2,0) |- (chart-4-1)
     node[near start,sloped,above] {};
    
     \draw
    (chart-1-1) edge (chart-2-1)
    \foreach \x/\y in {2/3, 3/4, 4/5} {
      (chart-\x-2) edge   (chart-\y-2) }
   ;  
  
  \end {tikzpicture}
\end{turn}
\caption {\label{fig:BSFB} Cascade of two Jaynes-Cummins open systems driven by a laser input in feedback configuration.}
 \end {figure}
\section{QNET}
We have reviewed quantum stochastic calculus in the appendix and here we discuss an algebra based on it. Before that we want to make sure that the logic of the observables of quantum stochastic calculus is a sub logic to a standard one which we accomplish via von Neumann algebras. The quantum walk evolution described above is a discrete-time dynamics and that of QNET algebra are time-continuous flows and so it is important to make sure that we work with standard logics.

\begin {prop} The logic of Weil operators of QNET *-algebra is a sub logic of standard logic. 
\end {prop}
\begin {proof}
Let us construct the von Neumann algebra of the Weil observables of the QNET *-algebra as $\mathscr{U} = vN(\mathscr{S}) = (\mathscr{S}\cup\mathscr{S}^*)^", \mathscr{S}\subset \mathscr{B}(\mathscr{F})$ using the double commutant theorem.
Let $\mathscr{U}$ be the von Neumann algebra of the observables of the quantum stochastic calculus. Then, 
\begin {equation}
\mathscr{L}_\mathscr{U} = \{ M: M \in \mathscr{L}(\mathscr{F}), P^M \in \mathscr{U} \}.
\end {equation}
is a logic with linear manifolds as elements that are ranges of projections of symmetric Fock space. $\mathscr{L}_\mathscr{U}$ is easily verified as a logic and a sub logic of the standard logic of symmetric Fock space. $\blacksquare$
\end {proof}

The QNET *-algebra has three operations, namely the serial connection $\triangleleft$, the parallel connection $\boxplus$, and the feedback to connect quantum systems. This algebra can be represented as Horn clauses using the AND, OR, and recursion connectives. An example of clauses to describe the above circuit ($\text{\color{blue}Figure \ref {fig:BSFB}}$) with $A_t(\Omega), A^\dag_t(\Omega), \Lambda_t(\Omega)$ as the three standard vacuum $(\Omega)$ quantum noises, whose discrete versions will play a role in describing quantum walks later, with Weyl predicates is shown below that drive the evolution:
\begin {align*}
\overset{3}{S}LH_{WBS}^t &\multimapinv \overset{3}{W}(t, \alpha, port_0) \boxtimes \overset{1}{b}s(t, bsp_i^1, bsp_i^2, bsp_o^1, bsp_o^2) \boxtimes 0-(bsp_i^1 = port_0). \\
\overset{3}{S}LH_{BSJC}^t &\multimapinv \overset{3}{S}LH_{WBS}^t \boxtimes \overset{1}{J}C_1(t, j1p_i^1, j1p_i^2, j1p_o^1, j1p_o^2) \boxtimes 0-(bsp_o^1 = j1p_i^1). \\
\overset{3}{S}LH_{QNET}^t &\multimapinv \overset{3}{S}LH_{BSJC}^t \boxtimes \overset{1}{J}C_2(t, j2p_i^1, j2p_i^2, j2p_o^1, j2p_o^2) \boxtimes 0-(j1p_o^1 = j2p_i^1). \\
\overset{3}{S}LH_{QNET}^t &\multimapinv \overset{3}{S}LH_{QNET}^t \boxtimes \overset{1}{(j}2p_o^2 = bsp_i^2). \\
\overset{3}{a}eSLH_{QNET}^t &\multimapinv \overset{3}{S}LH_{QNET}^t \boxtimes \overset{2}{r}educe_{SLH}.
\end {align*}
These time dependent predicates describe the Weyl unitary operators that are solutions to Hudson-Parthasarathy quantum stochastic differential equations. The SLH parameters that result in each step and the ports for the head predicates are not shown to focus on the connectivity between the components. The last clause will apply the adiabatic elimination theorem after verifying the required premises via the reduce predicate, assumptions one to four, and construct the SLH parameters accordingly. There are other model reduction techniques \cite {Luc2008} that can be applied, manually a cumbersome process, that can be automated within this framework. Thus, we have compact expressions of the circuit dynamics and a powerful way to reason about them.

Evans-Hudson flows $j_t(X)$ \cite {KP1995} are stochastic processes in Heisenberg picture, continuous version of $j_n$ quantum walks, that satisfy Hudson-Parthasarathy quantum stochastic differential equations and the corresponding predicates are dual to Weyl predicates.

\begin {lem}
The triple $(Weyl_t, HE_t, \iota)$ form a Chu space where $Weyl_t$ are predicates corresponding to second quantized operators satisfying Hudson-Parthasarathy quantum differential equations and $HE_t$ are the dual predicates corresponding to operators satisfying the Evans-Hudson flows.
\end {lem}

\begin {lem}
The triple $(Weyl_t, Weyl^{\bot}_t, \bot)$ form a Chu space where $Weyl_t$ are predicates corresponding to second quantized operators satisfying Hudson-Parthasarathy quantum differential equations and the relation $Weyl^{\bot} \wedge Weyl^{\bot}_t = 0$ satisfied. This Chu construction ensures backward chaining corresponds to refutations and forward chaining results in proofs. In the above if we replace the equality with an isomorphism $Weyl^{\bot} \wedge Weyl^{\bot}_t \rightarrow 0$ we will get a 2-Chu construction as it is a 2-morphism.
\end {lem}

\begin {exmp} In the context of logic programming an algorithm can be described as a theorem to be proved and here we describe the Polya{'}s urn scheme \cite {KP1992} in terms of quantized Horn clauses. It is an urn with a white balls and b red balls and a ball is drawn at random. It is replaced with c balls of same color and this process can be described by a classical Markov chain described in Appendix with equations \ref{eq: theta} of dimension 2 and \ref{eq: unitary}.\\

The state of the system is represented by $S=\mathscr{N}\times\mathscr{N}$.
\begin {align*}
\phi(x,y) &= (x + c, y), &\psi(x,y) &= (x, y + c). \\
p(x,y) &= x(x + y)^{-1}, &q(x, y) &= y(x + y)^{-1}.
\end {align*}
The corresponding Horn clauses can be set up as
\begin {equation}
(\theta^0_0{f})(x, y) \multimapinv (x + y)^{-1}\{xf(x + c, y) + yf(x, y + c)\}
\end {equation}
\end {exmp}
Then, the corresponding $j_n$ process can be used to get the probability distribution of (x,y). That is, the state of the urn is (x,y) will be true with a probability that can be found by directly measuring the observable.

Inspired by the approach in quantum filtering \cite{LUC2006}, where a commutative algebra is used for calculating the conditional probabilities, as in any single realization only compatible operators are measured, we propose a design for the theorem prover measuring only commuting observables in a single pass. A viable theorem prover will be based on an hybrid architecture of classical and quantum processing, for example, unification, resolution, and backtracking can be carried out by a classical processor. State preparation, measurements, and unitary evolutions can be performed at the quantum processor. 

\section{Conclusions and Open issues}
We have proposed a framework to reason with infinite quantum systems using ingredients from higher category theory but rooted in algebraic techniques. The proposed language has a focus on computation from logical deductions generalizing classical logical programming to quantum context. We have outlined few details for constructing a theorem prover as a quantum stochastic process that can be implemented on a quantum hardware. The framework treats countable classical states, classical-quantum interactions, time reversibility, unitary and more general completely positive map based evolutions. In addition, the framework accommodates interactions with environment in a natural way with the use of symmetric Fock space. The detailed treatment of semantics of the framework in terms of higher category theory will be investigated in the future. 

\section* {Acknowledgements}
The authors would like to thank the anonymous referees who provided very valuable feedback in the previous version of the manuscript. The author is grateful to Nikolas Tezak for the long discussions he had at Stanford, along with several cups of espresso, on QNET framework and developing the iPython scripts for the Jaynes-Cummins models.
\appendix

\section {Quantum probability}
The central ideas of classical probability consist of random variables and measures that have quantum analogues in self adjoint operators and trace mappings. The way probabilities are calculated differ significantly from classical theory and the events themselves manifest only after quantum mechanical measurements taken place.

\begin{defn} A finite dimensional quantum probability (QP) space is a tuple $(\mathscr{H},\mathbb{A}, \mathbb{P})$ where  $\mathscr{H}$ is a separable complex Hilbert space, $\mathbb{A}$  is a C* algebra that constitute the event space of orthogonal projections, and $\mathcal{P}$ is a trace class operator, specifically a density matrix in finite dimensional case, denoting the quantum state.\end{defn}
 
This is analogous to a classical probability (CP) space which is a tuple $(\Omega, \mathbb{F}, \rho)$ where $\Omega$ is the set of outcomes of a random experiment, $\mathbb{F}$ is the $\sigma$-algebra of events, and $\rho$ is a probability measure. 

\begin{defn} Canonical observables: Starting from a $\sigma$-finite measure space we can construct observables on a Hilbert space that are called canonical, as every observable can be shown to be unitarily equivalent to the direct sum of them \cite{KP1992}. Let $(\Omega, \Gamma, \mu)$, be a $\sigma$-finite measure space with a countably additive $\sigma-algebra$. We can construct the complex Hilbert space as a space of all square integrable functions w.r.t $\mu$ and denote it as $L^2(\mu)$. Then, the observable $\xi^{\mu}:\Gamma\rightarrow\mathcal{P}(\mathscr{H})$ can be set up as  $(\xi^\mu(E)f)(\omega)=I_E(\omega)f[\omega],f\in{L^2}(\mu)$ where I is the indicator function. 
\end{defn}
\begin{exmp} Let  $\mathscr{H}$=$\mathscr{C}^2$ and $\mathbb{A}=M_2$ the *-algebra of complex matrices of dimension $2 \times 2$ and the state  $\mathbb{P}(A)=\langle\psi,A\psi\rangle=\langle{A^\dag}\psi,\psi\rangle$ where $\psi$ is any unit vector. This space models quantum spin systems in physics and qubits in quantum information processing. This example can be generalized to n-dimensional space to build quantum probability spaces. 
\end{exmp}

\begin{defn} Quantum proposition: An atomic event which corresponds to a projection of a separable Hilbert space is a quantum proposition. As quantum observables can be decomposed into projections a canonical observable defined above can be treated as a proposition. A quantum predicate is a mapping that has bound variables which when instantiated to a specific value becomes a canonical variable such as the biased coin toss (BToss) described below.
\end {defn}
\begin{defn} Two quantum mechanical observables are said to be compatible, that is they can be measured simultaneously, if the operators representing them can be diagonalized concurrently. The two operators that share a common eigenvector will be characterized as co-measurable in terms of lattice elements in the section on logic programming.\end{defn}

\section{Entanglement semantics}

We base quantum probabilistic logic on separable Hilbert spaces of H-interpretations preserving the orthomodular lattice structure in accordance with the usual formulations of quantum logic \cite{VAR}. We describe closed formulas of Horn clauses using hermitian operators for random variables and Dirac measures for deterministic propositions treating commuting and non-commuting operators in the same setting. In the quantum case closed formulas result after a measurement projecting on to a specific eigen value of an observable. We can immediately see that we cannot construct H-interpretations for non-commuting observables simultaneously. Such constraints can be placed in a natural way using the operator algebraic formulation of the WFF. 

 Separable Hilbert spaces inherit the non-Boolean algebra of vector spaces described above and thus are suitable for basing the formalism of quantum mechanics. The projection operators of a separable Hilbert space can be thought of as the propositions of quantum logic as well as events of a probability space on which we can define quantum probability measures as established by Gleason.  The lattice on Hilbert spaces is non-Boolean as it violates the distributive law. The join and meet operations of this lattice are the set theoretic union (of closed linear spaces) and intersection respectively. From Figure 2 we have, $M'\boxtimes(M\vee{M}^{\perp})=M^o$. Whereas, $M^o\boxtimes{M}=M^o\boxtimes{M}^\perp=0$.
\\
\begin{figure}
\center{
\begin{tikzpicture}[
    scale=5,
    axis/.style={very thick, ->, >=stealth'},
    important line/.style={thick},
    dashed line/.style={dashed, thin},
    pile/.style={thick, ->, >=stealth', shorten <=2pt, shorten
    >=2pt},
    every node/.style={color=black}
    ]
    \draw[axis] (-0.1,0)  -- (1.1,0) node(xline)[right]
        {$M$};
    \draw[axis] (0,-0.1) -- (0,1.1) node(yline)[above] {$M^\perp$};
    \draw[important line] (0,0) coordinate (A) -- (.85,.85)
        coordinate (B) node[right, text width=5em] {$M^o$};
\end{tikzpicture}\\
\caption{The lattice of projective geometries that forms a non-Boolean algebra.}
}
\end{figure}
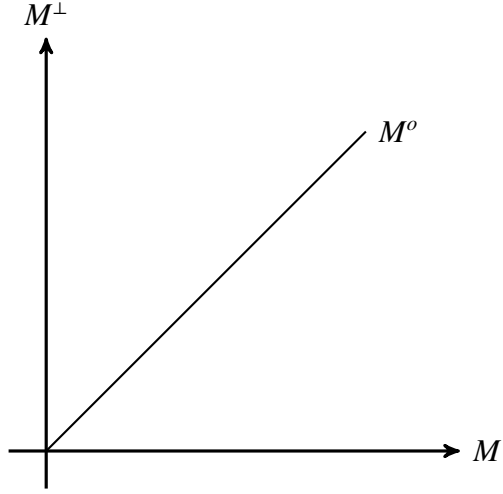

We can now extend the CP for probabilistic logic to a quantum context by defining the following Hilbert space, inner product and von Neumann algebras \cite{RB2015}: 
\begin{align*}
\mathscr{H}_{ES} & = L^2(\Omega = \cup_{j=1}^{\infty}Supp(I_j), \sigma(\Omega_F),\rho).\\
\langle{f}, g\rangle & =\int_\Omega {(f^*g)d\rho};f,g\in{H}.\\
\mathbb{A}_{ES}& =L^\infty(\Omega = \cup_{j=1}^{\infty}Supp(I_j),\sigma(\Omega_F),\rho)
\end{align*}

\begin{prop} The Hilbert space $\mathscr{H}_{ES}$ is separable. \\
\begin{proof} Let us define indicator functions on H-interpretations with exactly one atom being true. This forms a countably infinite basis and so $\mathscr{H}_{ES}$ is separable. Another way to look at is that this Hilbert space is isomorphic to the Banch space $l^2(\mathscr{N})$, the space of square-summable sequences. The *-algebra $\mathbb{A}_{ES}$ defined by point-wise multiplication of the measurable functions is commutative.
\end{proof}
\end{prop}

\textbf{Non-commutative version:}

The *-algebra $\mathbb{A}_{ES}$ considered above is commutative as the observables are compatible. In general, a set of Horn clauses may contain incomaptible observables and so we would like to have a *-algebra of $\mathscr{H}_{ES}$ that is non-commutative. One way to construct the von Neumann algebra of incompatible operator is to build a double commutant of bounded linear operators.
\begin{defn} A commutant of a set of operators A is defined as $A' = \{C: C\in{\mathbb{L}}^\infty(\mathscr{H}), [DC,CD]=0,\forall{D}\in{A}\}$. This is the set of bounded linear operators of the Hilbert space that commute with every member of the set A and this set need not be abelian.
\end{defn}
It can be shown that \cite{LUC2006} the double commutant of a self-adjoint set $\mathscr{S}$, that is $S\in \mathscr{S}\Rightarrow{s^*}\in\mathscr{S}$, forms a von Neumann algebra. The non-commutative quantum probabilistic space is defined below:
\begin{align*}
\mathscr{H}_{ES} & = L^2(\Omega = \cup_{j=1}^{\infty}Supp(I_j), \mathscr{A}_{ES},\rho).\\
\mathbb{A}_{ES}& =\mathscr{L}^\infty(\mathscr{H}_{ES}).
\end{align*}
For the non-commutative version the Hilbert space consists of H-interpretations of the first order language that has deterministic predicates as well quantum mechanical self adjoint operators as only Hermitian operators have spectral resolution and thus can be part of a lattice. As we cannot define a single probability space of incompatible operators we only have a von Neumann algebra and the quantum state in the definition. The Hilbert space is strictly not required in defining a quantum probability space but in this instance identifying a specific space helps develop intuition. The *-algebra $\mathscr{A}_{ES}$ of bounded linear operators will induce a lattice that is non-Boolean leading to difficulties in defining the AND operation. One way to understand the issue is in terms of joint probabilities that cannot be defined for incompatible observables as discussed in \cite{LUC2006}. In the quantum logic perspective let us discuss an example from \cite{KS} of a spin one half system.

\begin{exmp} The Hilbert space is $\mathbb{C}^2$ with the canonical basis  $e_1 = \ket{0}=\left( \begin{array}{c}  0 \\ 1 \end{array} \right)$ and $e_2 = \ket{1}=\left( \begin{array}{c}  1 \\ 0 \end{array} \right)$. We can define a lattice using the entire Hilbert space as the maximal element, null space as the minimal elements, and the projections corresponding the vectors $\ket{0}, \ket{1},E_{p'},$ and $E_{q'}$ as elements p, q, p', and q' respectively. Here, p' and q' are two vectors at an angle $\pm\theta$ to the vector $\ket{0}$. Let us also assume that the spin one half system is prepared in the pure state $\ket{0}$. The projections $E_{p'}$ and $E_{q'}$ are incompatible and the lattice as depicted in Figure 1 is non-Boolean. The join and meet binary operations of this lattice are defined as the set theoretic union and intersections of the range of the projections. Suppose, the angle $\theta$ is chosen in such a way that the probability of the projections $E_{p'}$ and $E_{q'}$ is the same as 0.9999999. We can see that the probability of their meet is zero as is clear from the Hasse diagram. Even though the individual propositions are highly probable, their joint probability is undefined. \end{exmp}

In the non-commutative versions, the canonical observables defined above can be used to construct the lattice of the quantum logic. By spectral theorem, a self-adjoint operator can be decomposed into projection operators, and the range of the operators can be used to define the operations of the lattice. The implied logical operation is defined in terms of the inclusion of the range sets. The join operation $(\vee)$ is the set theoretic union, and the meet $(\boxtimes)$ is the set theoretic intersection acting on ranges (closed subspaces). Thus, it is perfectly valid to construct Horn clauses in terms of quantum mechanical operators.

\begin {defn} Conditional expectation: Let ($\mathscr{N}, \mathbb{P}$) be a quantum probability space with $\mathscr{A}\subset\mathscr{N}$ be a commutative *-sub algebra. Then the mapping $\mathscr{P}(.|\mathscr{A}):\mathscr{A'}\rightarrow\mathscr{A}$ is a
conditional expectation if $\mathscr{P}(\mathscr{P}(B|\mathscr{A})A) = \mathbb{P}(BA) \forall{A}\in\mathscr{A}, B\in\mathscr{A'}$. It is easy to show that this definition of conditional expectation satisfies usual properties such as tower and projections. In the quauntum context conditional expectation is not always possible and whenever we can define one such operation the Horn clauses can be set set up as follow:
\begin {align*}
P[D|A] &\multimapinv \sum_i\frac{P(DA_i)}{P(A_i)}.\\
P[D|A] &\multimapinv (D\in{A'})\boxtimes\sum_i\frac{P(DA_i)}{P(A_i)}A_i\boxtimes{spec(A)}={A_i}\\
P[D|A] &\multimapinv (D\in{A'})\boxtimes\sum_i\frac{P(DA_i)}{P(A_i)}A_i, A_i\in{spec(A)}
\end {align*}
\end {defn}

\section{Quantum logic}

The quantum logic framework consists of key theorems that relate observables, quantum states to classical probability spaces obviating Born rule of calculating the probabilities from the axioms of quantum mechanics. The links between the theorems can be summarized as \cite{harding} canonical observables $\zeta$ mapping Borel sets of real line or $\eta$ discrete measure space to  projections of a separable Hilbert space $\mathcal{P}(\mathscr{H})$ to probability measures via Gleason's theorem.
\begin {align*}
\zeta & :\mathbb{B}(\mathscr{R}) & \rightarrow \text{   }\mathcal{P}(\mathscr{H});\text{    }\mu & \rightarrow {(0, 1)}.  \\
\eta & :\Sigma(\mathscr{N}) & \rightarrow \text{   }\mathcal{P}(\mathscr{H});\text{    } \mu & \rightarrow {(0, 1)}.
\end {align*}
It can be shown rigorously that every quantum mechanical observable can be decomposed into canonical observables via Hahn-Hellinger{'}s theorem \cite{HAHN} , self adjoint operators have spectral decompositions on orthogonal projections of a Hilbert space using spectral theorem for operators, and finally every probability measure on the orthomodular lattice of projections of a Hilbert space would correspond to a quantum state via Gleason's theorem. The dynamics of a quantum system can be described using Wigner's theorem \cite{VAR} that map evolutions as automorphisms of unitaries and anti-unitaries of the Hilbert space. The above arrows can be reversed to go from Schr$\ddot{o}$dinger to Heisenberg picture and along with Wigner's theorem completes the description of quantum mechanics from the perspective of quantum logic. On the left the Borel sets form a Boolean lattice while on the right the lattice of projections could possibly form a non Boolean lattice that violates the distributive law. Multiple observables, random variables, can be defined from the Borel sets or a discrete measure space to the lattice of projections. Different observables can define probability measures on the projections with different supports, that may have some overlap, on the lattice giving rise to non Boolean structure. This results in non compatible observables that can't be measured simultaneously with arbitrary precision. Canonical observables could have continuos or discrete domains and ranges in projection valued measures of a Hilbert space. Quantum mechanical harmonic oscillator is an example where this can be done that lead to observables with Poissonian as well as gaussian distributions with only one type of noise in each realization depending upon the selected basis for measurement. 
\begin {defn} An automorphism is a map $\tau:\mathcal{P}(\mathscr{H})\rightarrow{\mathcal{P}(\mathscr{H})}$ that is one-to-one and onto satisfying the following conditions:
\begin {align*}
\tau(0) & = 0. \\
\tau(1) & = 1. \\
\tau(\vee_j{E_j}) & = \vee_j\tau({E_j}). \\
\tau(\boxtimes_j{E_j}) & = \boxtimes_j\tau({E_j}). \\
\tau(1 - E) & = 1 - \tau(E).
\end {align*}
\end {defn}
These maps conserve the structures of orthomodular lattices that lead to the definition of standard logics that are used in this work.
\begin {defn} A standard logic is a system with an underlying lattice isomorphic to a lattice of closed linear manifolds of an infinite dimensional separable Hilbert space. The set inclusion operation defines the partial order of the lattice. Consequently, quantum evolutions either in Schr$\ddot{o}$dinger or Heisenberg picture, will conserve the lattice structure allowing us to use unitaries freely in the construction of Horn clauses. 
\end {defn}
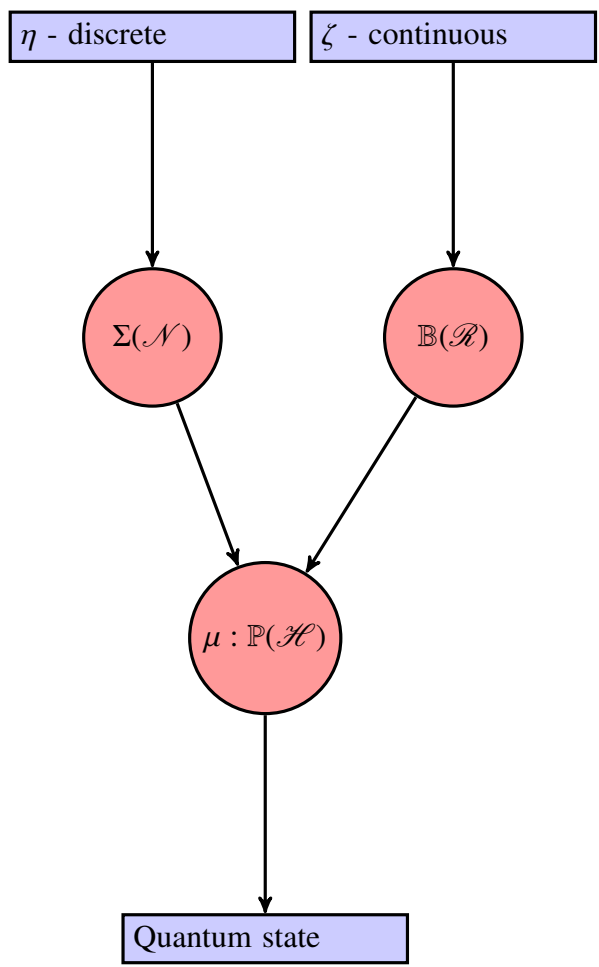
\begin {figure}
\centering
\begin{tikzpicture}[
    and/.style={and gate US,thick,draw,fill=red!60,rotate=90,
		anchor=east,xshift=-1mm},
    or/.style={or gate US,thick,draw,fill=blue!60,rotate=90,
		anchor=east,xshift=-1mm},
    be/.style={circle,thick,draw,fill=green!60,anchor=north,
		minimum width=0.2cm},
    tr/.style={buffer gate US,thick,draw,fill=purple!60,rotate=90,
		anchor=east,minimum width=0.8cm},
    label distance=5mm,
    every label/.style={blue},
    event/.style={rectangle,thick,draw,fill=yellow!20,text width=2cm,
		text centered,font=\sffamily,anchor=north},
    edge from parent/.style={very thick,draw=black!70},
    edge from parent path={(\tikzparentnode.south) -- ++(0,-1.05cm)
			-| (\tikzchildnode.north)},
    level 1/.style={sibling distance=7cm,level distance=1.4cm,
			growth parent anchor=south,nodes=event},
    level 2/.style={sibling distance=7cm},
    level 3/.style={sibling distance=6cm},
    level 4/.style={sibling distance=3cm}
    ]
   \begin{scope}[xshift=-7.5cm,yshift=-5cm,very thick,
		node distance=4.0cm,on grid,>=stealth',
		block/.style={rectangle,draw,fill=blue!20,text width=3.5cm},
		comp/.style={circle,draw,fill=red!40,minimum size=52pt}]
   \node [block] (re)					{Quantum state};
   \node [comp]	 (cb)	[above=of re]			{$\mu:{\mathbb{P}(\mathscr{H})}$}  edge [->] (re);
   \node [comp]	 (ca1)	[above=of cb,xshift=-1.5cm]	{$\Sigma(\mathscr{N})$} edge [->] (cb);
   \node [comp]	 (ca2)	[right=of ca1, xshift=0.0cm]	{      $\mathbb{B}(\mathscr{R})$      } edge [->] (cb);
   \node [block] (s1)	[above=of ca1]		{$\eta$ - discrete} edge [->] (ca1);
   \node [block] (s2)	[right=of s1]		{$\zeta$ - continuous} edge [->] (ca2);
   \end{scope}
    \end{tikzpicture}
   \caption {Construction of quantum probability a space for the harmonic oscillator that consists of both discrete and continuous random variables. $\eta$ and $\zeta$ generate two von Neumann algebras $\Sigma(\mathscr{N})$ and $\mathbb{B}(\mathscr{R})$ respectively. These in turn lead to projection operator valued measures (POVMs) with overlapping support. The resulting probability measure on the orthogonal projections of the Hilbert space defines a quantum state due to Gleason's theorem. }
  \end {figure}
Due to the interdisciplinary nature of the formalism we use to describe quantum probabilistic logic programming it can be approached in several different ways. Here, we start the discussion with an introduction to lattices and build non Boolean logics of projections on a Hilbert space. We can then construct probability measures on them that lead to quantum states via Gleason's theorem discussed later in the section on entanglement semantics. 

\textbf{Lattices}: These mathematical structures, with additional  orthocomplement and modular properties, form the basis for logic that is classical as well as quantum and so we examine a few examples of them \cite{KS}. Any partially ordered set with a unique maximal and minimal element and the two binary operations meet and join define a lattice. The events structure of a probability space is based on a lattice and other examples of a lattice include propositional logic, Borel sets of the real line, any set and the associated power set, and a multi-dimensional vector space and its subspaces.
\begin{defn} An orthocomplement of an element a of lattice $\mathscr{L}$ is an element a' of $\mathscr{L}$ such that $a\vee{a'} = 1$ and $a\boxtimes{a'} = 0$.\end{defn} 
\begin{defn} A complemented, that is every element has a compliment, lattice $\mathscr{L}$ is modular if it satisfies 
the property $c < a \Rightarrow a\boxtimes({b} \vee {c}) = (a\boxtimes{b})\vee{c}, \forall{a,b,c\in{\mathscr{L}}}$.\end{defn}
\begin{defn} Two elements p and q of the lattice $\mathscr{L}$ are co-measurable if there exist mutually orthogonal propositions a,b, and c such that $p = a\vee{b}$ and $q = a\vee{c}$.\end{defn} Two mutually orthogonal propositions are co-measurable as the element `a` can be chosen to be the zero element. Two co-measurable propositions sharing a common element are equivalent to the eigenvector shared by the projection operators. Several lattices (Boolean) of co-measurable events constitute a quantum logical system. Whenever measurements are made the system can be described by a single Boolean lattice and the tools of classical logic can be applied to study the system. When dealing with unbounded observables such as the gaussian distributed quadratures used in quantum optics the lattice can be constructed out of the Borel sets of the real line. The lattices of the logic are defined on a separable Hilbert space and so the composition of them is well defined in terms of tensor product of Hilbert spaces. 												

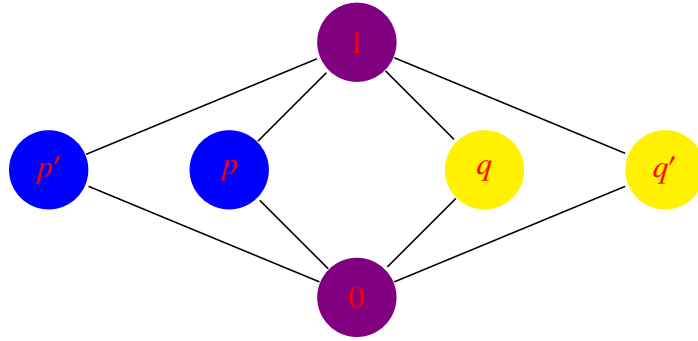
\begin{figure}
\center{
\begin{tikzpicture}[->,>=stealth',shorten >=1pt,auto,node distance=2.4cm,
                    semithick]
  \tikzstyle{every state}=[fill=blue,draw=none,text=red]

  \node [state, fill=blue]        (A)                    {$p$};
  \node[state, fill=blue!50!red]         (B) [above right of=A] {$1$};
  \node[state, fill=blue!50!red]         (D) [below right of=A] {$0$};
  \node[state, fill=yellow]         (C) [below right of=B] {$q$};
  \node[state, fill=blue]         (E) [left of=A]       {$p'$};
  \node[state, fill=yellow]         (F) [right of=C]       {$q'$};

  \path[-] (A)  edge              node {} (B)
            
        (B) edge node {} (F)
            edge              node {} (C)
        (C) edge              node {} (D)
           
        (D) edge  node {} (F)
            edge              node {} (A)
             edge              node {} (E)
        (E) edge  node {} (B);
\end{tikzpicture}
\caption{ The Hasse diagram of the logic with incompatible observables resembling a "Chinese lantern". This system is a composite of two boolean lattices with a common maximal and minimal elements.}
}
\end{figure}

 Let us start with an example of quantized Horn clauses by expressing the no-cloning theorem in our formalism.

\begin {exmp}
The no-cloning theorem \cite{ZUREK} of quantum information processing states that an unknown quantum state cannot be copied by proving the non existence of any unitary operator that would effect the cloning. Here we formulate the theorem and its proof using Horn clauses and trace the steps of a theorem prover. We use \\
\begin {align*}
\overset{2}{c}lone(\Psi, \Psi) &\multimapinv (state \overset{1}{=} \ket{\Psi}\otimes\ket{\Psi}), \overset{2}{U}(\Phi, state), state\overset{1}{=}(\ket{\Psi}\otimes\ket{\_}).\\
\alpha{\overset{2}{U}}(\ket{\Psi}\otimes\ket{\_},state) + \beta{\overset{2}{U}}(\ket{\Phi}\otimes\ket{\_},state) &\multimapinv {\overset{2}{U}(\alpha\ket{\Psi} + \beta\ket{\Phi}\otimes\ket{\_},state)}.\\
\alpha{\overset{2}{U}}(\ket{\_}\otimes\ket{\Psi},state) + \beta{\overset{2}{U}}(\ket{\_}\otimes\ket{\Phi},state) &\multimapinv {\overset{2}{U}(\ket{\_}\otimes\alpha\ket{\Psi} + \beta\ket{\Phi},state)}.\\
\overset{2}{U}(\Phi,state) &\multimapinv state\overset{1}{=} e^{\Phi}(\Phi). \\\text{   An arbitrary phase factor is added to state.}
\end {align*}

The first clause sets the rule that the clone predicate is true if a composite identical states are produced starting from a state$\Psi$ and a place holder for the second state. The second and third Horn clauses hold because of the linearity of the unitary operator U. The last clause defines the unitary operator. Suppose we have the goal:\\

$clone(\alpha\ket{0}+\beta\ket{1}, \alpha\ket{0}+\beta\ket{1})=clone(\alpha\ket{0}+\beta\ket{1},\ket{\_})$ \\

we expect it to succeed by the following resolutions:\\

L.H.S:
\begin {align*}
 state &= \alpha\ket{0}+\beta\ket{1}\otimes{\alpha\ket{0}+\beta\ket{1}}. \\
 state &= \alpha^2\ket{0}\ket{0}+\alpha\beta\ket{0}\ket{1} + \alpha\beta\ket{1}\ket{0}+\beta^2\ket{1}\ket{1}.\\
 L.H.S  &= \alpha^2{\overset{2}{U}(}\ket{0}\ket{0})+\alpha\beta{\overset{2}{U}(}\ket{0}\ket{1},state) + \alpha\beta{\overset{2}{U}(}\ket{1}\ket{0},state)+\beta^2{\overset{2}{U}(}\ket{1}\ket{1},state).
\end {align*}

On the other hand the R.H.S is

 $\overset{2}{U}(\alpha\ket{0}\ket{0},state) + \overset{2}{U}(\beta\ket{1}\ket{1},state)$.\\

This leads to a contradiction as there are no $\alpha$ and $\beta$ values that satisfy both the equations. In an actual technical implementation a theorem prover starts with a negation of the above goal and derive a contradiction which is not possible in this case.
\end {exmp}

\section{Time-discrete quantum stochastic processes}

Let $\{H_n, n\geq0\}$ be a countable sequence of Hilbert spaces with 
\begin{equation}
\{\phi_n,n\geq1,\parallel{\phi_n}\parallel=1,\phi_n\in{H_n}\}
\end {equation}
  as the stabilizing sequence of unit vectors with respect to which countable tensor products are defined. The establishing sequence is required to define the inner product of the composite space as
\begin{equation}
M=\{\underbar{u}|\underbar{u}=(u_1,u_2,\dots), u_n\in\mathscr{H}_n, u_n=\phi_n \text{   for all large n}\}
\end{equation}
\begin{equation}
K(\underbar{u},\underbar{v})=\prod_{j=1}^\infty{\langle}u_j,v_j\rangle, \underbar{u}, \underbar{v}\in{M}
\end{equation}
 This sequence preventing the inner product from growing without bounds will contain the interpretation of the state under which the evolution of the system occurs in the Heisenberg picture. It is usually taken to be the ground state of the system and we shall adopt the same convention in this work and set up compact notations for composite spaces, for example to denote the sequence of Hilbert spaces where events later than time instance n occur can be represented as: 
\begin{equation}
H_{[n+1}=H_{n+1}\otimes{H_{n+2}}\otimes{\dots}\text{, w.r.t. } \phi_{n+1},\phi_{n+2},\dots
\end {equation} 
Similarly, the sequence of past will be denoted by the subscript ] and with this notation the total Hilbert space can be split as follows:
\begin{equation}
H=H_0\otimes{H_{[1}}=H_{n]}\otimes{H_{[n+1}}
\end{equation}
 Similarly, we can define an increasing sequence of algebras from the bounded operators of the respective Hilbert spaces as follows:
\begin{equation}
B_n=B_{n]}=\{X\otimes{1_{[n+1}},X\in{B_{n]}}\} \text{   and   }B_{\infty}=B(H)
\end{equation}
The sequence of algebras $B_n$ may be thought of as information about the flow up to the point n (current state).
The conditional expectation given the current information about the flow is defined as a map $E{_n]}:B_{\infty}\rightarrow{B_{n]}}$ satisfying  
${\langle}u,E_{n]}(X)v\rangle={\langle}u{\otimes}\phi_{[n+1},Xv{\otimes}\phi_{[n+1}\rangle,u,v\in{H_{n]}} \text{  and   }X\in{B_\infty}$. It follows that
\begin{equation}
E_{n]}(X)=E_{\ket{\phi_{[n+1}}\bra{\phi_{[n+1}}}(X)\otimes{1_{[n+1}}
\end{equation}
By taking a conditional expectation we are projecting the operator X to the past, that is the algebra $B_n]$ by averaging over the future algebra $B_{[n+1}$ which is signified by the expectation with respect to the state $\ket{\phi_{[n+1}}\bra{\phi_{[n+1}}$. 
If $\rho$ is a state in $H_2$ then an operator $E_{\rho}(Z)\in{B(H_1)}$ where $Z\in{B(H_1\otimes{H_2})}$ is defined as
\begin{equation}
\langle{u},E_{\rho}(Z)v\rangle=trZ(\ket{v}\bra{u}\otimes{\rho}),\text{  u,v}\in{H_1}
\end{equation}
We now consider a quantum flow, a sequence of operators whose existence can be established using induction. The flow is induced by a homomorphism $\Theta:B(H)\rightarrow{B(H_1)}$ that is composed of a set of maps $\theta_j^i:E_{\ket{e_j}\bra{e_i}}\Theta(X)$. Here, $H_1$ is another Hilbert space with the basis, $\{e_0, e_1, \dots, e_d\}$ and we choose the stabilizing sequence $\phi_n = e_n$ for the tensor product.
\begin{prop}\cite{KP1992}
 There exists an operator valued process $J_n:B_0\rightarrow{B_{n]}}$ satisfying 
\begin{equation}
J_0(X)=X\otimes{1_{[1}},\text{    }\Theta(X)=X\otimes{1_{[2}}
\end{equation}
\begin{equation}
J_n(X) = \sum_{0\leq{i,j}\leq{d-1}}j_{n-1}(\theta_i^j(X))1_{n-1]}\otimes{\ket{e_i}\bra{e_j}}\otimes{1_{[n+1}}
\end{equation}
In addition, the following Markov property holds: 
\begin{equation}
E_{n-1]}J_n(X)=J_{n-1}(\theta_0^0(X)),\forall{X\in{B_0}},n\geq{1}
\end{equation}
\end{prop}
The map $\theta_0^0$ encodes the classical states, and the transition probabilities of a Markov chain and the Hilbert space $H_1$ play the role of the coin space in the usual formulations of quantum walks \cite{SV}. Furthermore, the $\theta_j^i$ can be shown to be involutive, unital, linear, and positive maps \cite{KP1992}. 

Let us now look at a process to construct a quantum stochastic flow from a classical Markov chain with the transition probability matrix $T  = ( p_{ij} )$ \cite{KP1992}: 
The classical Markov chain is defined on a probability space $(S,F,\mu)$ where S is a finite set of cardinality d and $\phi_i:S\rightarrow{S} $ are measurable maps. These measurable functions provide flexibility in defining the transition probability via multiple maps or degrees of freedom. In the quantum analogue, the classical states are defined on the Hilbert space $ H_0 = L^2(\mu)$, the space of all square integrable functions w.r.t $\mu, B_0 = L^\infty(\mu)$ the space of bounded measurable functions, and $H_1 = C^d$ (d-dimensional complex space), and $B_1 = B(H_1)$. We can construct a homomorphism to induce the flow as follows:

\begin{equation}  \label{eq: theta}
\Theta(f)=((\theta^0_0(f)))=U\left( \begin{array}{ccc}
             f\circ\phi_0 & { }  & 0\\
             {} & f\circ\phi_1 &{}\\
             0  & {} & \ddots{f}\circ\phi_{d-1}
              \end{array} \right)U^*
\end{equation}\\

\begin{equation} \label{eq: unitary}
U=\left( \begin{array}{cccc}
             p_0^{1/2} & p_1^{1/2 }  & \dots &p_{d-1}^{1/2}\\
            -p_1^{1/2}  &{ }  &{ }  &{ }\\
             \vdots &1  &-  &Q\\
            -p_{d-1}^{1/2}  &{ } &{ } &{ }
             \end{array} \right) \text{, where   } Q=((q_{ij})),q_{ij}=(p_{i}p_j)^{1/2}(1+p_0^{1/2}), i,j\geq{1}
\end{equation}

\textbf{Composite spaces:} In classical probability, composite spaces are represented by product spaces, and, in the case of QP, the tensor product is used to combine the systems. Suppose,  $(\mathscr{H}_1,\mathbb{A}_1,\mathbb{P}_1)$ and  $(\mathscr{H}_2,\mathbb{A}_2, \mathbb{P}_2)$ are two quantum probability spaces.Then the total system is represented by  $(\mathscr{H}_1\otimes\mathscr{H}_2,\mathbb{A}_1\otimes\mathbb{A}_2, \mathbb{P}_1\otimes\mathbb{P}_2)$.
Let us develop some intuition on tensor products of Hilbert spaces in statistical terms. The essential ingredient to construct a Hilbert space that is a tensor product is to define the inner product of the new space which leads to the following definition and a very important result \cite{KP1992}:

\textbf{Positive definite kernel:} Let $\mathscr{H}$ be any set, possibly countably infinite. A positive definite kernel on $\mathscr{H}$ is a complex number valued map $K:\mathscr{H}\times$$\mathscr{H}\rightarrow$$\mathscr{C}$ satisfying \\
\begin{equation}
\sum\limits_{i,j}\bar{\alpha_i}\alpha_j{K(x_i,x_j)}\geq{0}, \alpha_i\in\mathscr{C},x_i\in\mathscr{H}.
\end{equation}
The inner product of a Hilbert space is an example of a positive semidefinite kernel. Given an arbitrary positive semidefinite kernel K, we can construct a Hilbert space and its inner product in terms of K \cite{KP1992}. By exploiting the fact that the product of two kernels (matrices) is also kernel, and that an n-dimensional kernel can be represented by the covariance matrix of 2n identically distributed Gaussian random variables, we can see that the inner product of the composite system is the covariance of random variables. This classical correlation, Gaussian covariance structure, is violated when a composite quantum system is entangled. The statistical interpretation of a tensor product is rooted in the Gaussian measure of the composite space. The extension of this to infinite tensor products is enabled by Kolmogorov's consistency theorem. That is, all finite subsets of $\mathscr{H}$ give rise to Gaussian measures satisfying the consistency conditions. Hence, the existence of a family of Gaussian measures indexed by the members of $\mathscr{H}$ is guaranteed. In other words, the inner product of the countable tensor product of Hilbert spaces is well defined and is given by the following definition:\\
\begin{defn} Let $\mathscr{H}_n,\phi_n$ be a sequence of Hilbert spaces and unit vectors such that $\phi_n\in{\mathscr{H}_n},\parallel{\phi_n}\parallel=1$. The inner product of the composite space is defined w.r.t this stabilizing sequence as
\begin{equation}
M=\{\underbar{u}|\underbar{u}=(u_1,u_2,\dots), u_j\in\mathscr{H}, u_n=\phi_n \text{   for all large n}\}.
\end{equation}\end{defn}
\begin{equation}
K(\underbar{u},\underbar{v})=\prod_{j=1}^\infty{\langle}u_j,v_j\rangle, \underbar{u}, \underbar{v}\in{M}.
\end{equation}
The requirement for the Kolmogorov consistency theorem manifests as the stabilizing sequence in defining the inner product. It is important to note that the countable tensor product of Hilbert spaces is not a separable Hilbert space.

\section {Quantum Stochastic Processes}
The stochastic processes in the quantum context in this work are Poisson processes and a pair of conjugate Brownian motions that form classes of non commuting Hermitian operators.
\begin {defn} Poisson process on symmetric Boson Fock space: Fock space is a Hilbert space as defined below as the set of square integrable functions on a space with respect to a measure that assigns equal probabilities to jump times $t_1, t_2, \dots, t_n$ as would be expected of a poisson process.  
\begin {align*}
\Omega &= \cup_n{\Omega_n}; \\
\Omega_0 &= \emptyset; \\
\Omega_n &= \{t_1,t_2,\dots,t_n\}; t_1 < t_2 < \dots < t_n\in{[0, T]}.
\end {align*}
\begin {equation*}
P_n(\Omega_n) = \frac{e^{-T}T^n}{n!}.
\end {equation*}
\begin {equation*} \label {FockEq}
\mathscr{H} = L^2\big(\omega,\mathscr{F},\rho\big).
\end {equation*}
\begin {equation*}
\mathscr{W} = B(\mathscr{H}). \text {Bounded linear operators}
\end {equation*}
This space has a continuous tensor product structure as shown below that facilitates defining time-continuous stochastic processes.
\begin {align*}
\Omega_{[s,t]} &= \Omega_{s]}\otimes\Omega_{[s,t]}\otimes\Omega_{[t}, s,t\in[0, T].\\
\mathscr{F}_{[s,t]} &= \mathscr{F}_{s]}\otimes\mathscr{F}_{[s,t]}\otimes\mathscr{F}_{[t}.\\
\mathscr{W}_{[s,t]} &= \mathscr{W}_{s]}\otimes\mathscr{W}_{[s,t]}\otimes\mathscr{W}_{[t}.
\end {align*}
Let us now define special vectors called coherent vectors that also factorize continuously in time and their inner product as:
\begin {align*}
e(f)(\emptyset) &= 1. \\
e(f)(\tau) &= \prod_{t\in\tau}^\infty {f(t)}; f\in{L^\infty([0,T])}. \\
\langle{e(f)}, e(g)\rangle &= e^{\norm{f}^2_2-T}.\\
e(f) &= e(f_{s]})\otimes{e(f_{[s,t]})}\otimes{e(f_{[t})}.
\end {align*}
\end {defn}
Let us denote by D the linear span of the exponential vectors that form the domain of the operators of the stochastic processes.
It is enough to define operators with domain as D as it is dense in the Hilbert space the operations are uniquely defined. The Hilbert space we have constructed is on a classical probability space and so we can define Poisson random variables. The random variable $N_t(\tau) = \abs{\tau\cap[0,t]}$ counts the number of jumps up to time t and it is a Poisson process with unit rate under the probability measure P. An operator process can be constructed from this as follows:
\begin {equation}
(\Lambda_t\Psi)(\tau) = N_t(\tau) = \abs{\tau\cap[0,t]}\psi(\tau), \Psi\in\mathscr{F}, \tau\in\Omega, t\in[0,T].
\end {equation}
We can now define a family of states called coherent states as $\mathbb{P}_f(X) = \langle{e(f)}, Xe(f)\rangle^{T-\norm{f}^2_2}$.
The physical intuition is clear as these states describe the coherent states of quantum optics. We will designate $\mathbb{P}_0=\emptyset$ as the vacuum state and $e(0) = \Phi$ as the vacuum vector. The process $\{\Lambda_t\}$ defined is called the gauge process which is commutative and later after defining the Brownian motions we will construct a non commutative version that will form part of the quantum noises. We have ignored technicalities such as affiliated process that are analogous to adapted processes in classical system and refer the readers to \cite {LUC2006} for details.

\begin {defn} Brownian motions: We continue to work with the Fock space defined on a continuum as above and define a Weyl operator as
\begin {equation}
W(f)e(g) = e^{-\int_0^T{(f^\dag(t)g(t) + \frac{1}{2}f^\dag(t)f(t))dt}} e(f + g) = e^{-\langle{f,g}\rangle - \norm{f}^2_2}e(f + g).
\end {equation} 
\end {defn}
This is a unitary operator that may be considered as a second quantized on the Fock space and can be shown to posses the continuous tensor property. Let us fix $f\in{L^\infty([0, T])}$ and construct the group of unitary operators $\{W(tf)\}_{t\in{R}}$ and by Stone's theorem guarantee's a self adjoint operator B(f) such that $W(kf) = e^{kB(f)}$. This is similar to constructing a unitary out of a Hamiltonian and in this context these operators are called field operators. It is easy to establish that the probability distribution of these random variables is gaussian in the coherent state $\mathbb{P}_g$. Like in the case of the gauge process we can construct another operator process as $\{B^\phi_t = B(e^{i\phi}\chi[0,T]): t\in[0, T]$ for some fixed real function $\phi\in{L^\infty}([0,T])$ that has Gaussian probability law at each time epoch. Now, we can define the following pair of conjugate Brownian motions A and $A^\dag$:
\begin {align*}
Q_t &= B(i\chi[0,T]). \\
P_t  &= B(-\chi[0,T). \\
A_t &= \frac{(Q_t + iP_t)}{2}. \\
A^\dag_t &= \frac{(Q_t - iP_t)}{2}.
\end {align*}
The process $\Lambda_t$ vanishes in vacuum coherent state and so let us define the process $\Lambda_t(f) = W(f)^\dag\Lambda_t{W(f)}$ that has same statistics in state $\mathbb{P}_f$ as $\Lambda_t$ in vacuum state. The sandwiching with a Weyl operator (defined below) provides a way to transform statistics of a coherent state to vacuum which could be the choice to carry out all the analysis. The three processes $(A^\phi_t, A^{\dag\phi}, \Lambda^\phi_t)$ are called the quantum noises that can be used to describe the dynamics of an open quantum system such as a systems coupled to a Bosonic bath ($\text{\color{blue}Figure \ref {fig:OpenSystem}}$). These quantum noises are mathematical objects called white noises that are good approximations to wide band noises encountered in quantum optics \cite{Accardi1995}, \cite{Gough2005}. The dynamics of an open system interacting with an environment can be described by a quantum stochastic differential equation (QSDE) of Hudson-Parthasarathy kind \cite{KP1992} as
\begin {equation} \label {eq:QSDE}
dU_t = \{({\color{red}S} - 1)d\Lambda + {\color{blue}L}dA^{\dag}_t  - {\color{red}S}{\color{blue}L}^{\dag}dA_t + (i{\color{green}H}-\frac{1}{2}{\color{blue}L^{\dag}L})dt \}U_t.
\end {equation}
In the above equation, the unitary operator U is defined on the combined system and the Fock space described by the \text{\color{blue} Equation \eqref {FockEq}}. The operator L, it is actually a vector with one element per noise channel, and its conjugate are the Lindbladians corresponding to the channels of decoherence, and the operator S is a similar noise channel that is of discrete in time. When the above equation is traced out with respect to the bath we obtain the quantum master equation with the operator S missing for the obvious reasons.  The three parameters S,L, and H of the QSDE characterize the open evolution of a system that are coefficients of Hudson-Parthasarathy quantum stochastic differential equations denoting the internal energy of the system in terms of the Hamiltonian H, couplings to the environment  via the Lindbladians, and the scattering by the fields by the S matrix. We refer the readers to the work of Gough and James \cite{Gough2008} for the details of the SLH mathematical framework that derives the QSDE for composite systems connected in a network. A computational implementation of the SLH framework can be found at the work of Tezak and Mabuchi \cite{Tezak2012} that has a user friendly interface to connect standard components in quantum optics such as beam splitters, optical cavities, and optical parametric oscillators to mention a few. The computational framework QNET provides numerical tools based on QuTip \cite {Nori2013} to solve the QSDEs constructed out of networked components of quantum optics. We will construct two examples based on Jaynes-Cummins models on the quantum probability space ($\mathscr{H}_0\otimes\mathscr{H}, \mathscr{B}(\mathscr{H}_0)\otimes\mathscr{W}, \rho_0\otimes\Omega$), where $\Omega$ is the vacuum state of the Fock space and $\rho_0$ is the initial state of the system, derive the QSDE \text{\color{blue} Equation \eqref{eq:QSDE}} for the combined system in Heisenberg picture and quantify the entanglement that results from the two connected cavities using concurrence.  

\begin {defn} Weyl operators: The three second quantized operators $(A^\phi, A^{\dag\phi}, \Lambda^\phi)$ are members of a representation of the Euclidian group over $\mathscr{H}$ \cite {KP1992} whose generic form is:
\begin {align*}
W(u, U)e(v) = \{exp(-\frac{1}{2}\norm{u}^2)-\langle{u},Uv\rangle\}e(Uv + u), \forall{v}\in\mathscr{H}.
\end {align*} 
as we will see later these Weyl operators used to glue together components of quantum optical circuits. 
\end {defn}
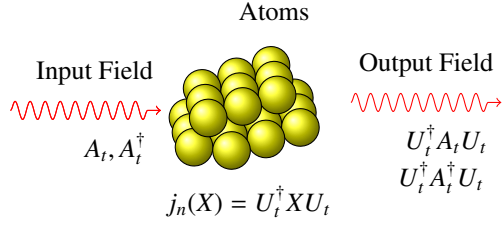
\begin{figure}
\centering
\begin{tikzpicture}
\begin{scope}[yshift=-180,yslant=0.5,xslant=-1]
\end{scope}

\begin{scope}[rotate around = {-5:(0,0,0)}]
    \begin{scope}[decoration={snake,amplitude=1.2mm,
        segment length=2mm,post length=1mm}] 
      \draw[decorate,red,->] (-1.5,-3) -- ++(5:2);
      \draw (-1.5,-2.5) ++(5:2) node[left] {\footnotesize Input Field};
      \draw (-1.5,-3.5) ++(5:2) node[left] {\footnotesize $A_t,A^{\dag}_t$};
       \draw[decorate,red,->] (-1.5,-3) -- ++(5:2);
       \draw[decorate,red,->] (3.0,-2.5) -- ++(5:2);
      \draw (3.0,-2.0) ++(5:2) node[left] {\footnotesize Output Field};
      \draw (3.0,-3.0) ++(5:2) node[left] {\footnotesize $U^{\dag}_tA_t{U_t}$};
       \draw (3.0,-3.5) ++(5:2) node[left] {\footnotesize $U^{\dag}_tA^{\dag}_t{U_t}$};
    \end{scope}
    
    \draw (0.5,-1.5) ++(5:2) node[left] {\footnotesize Atoms};
    \draw (1.0,-4.0) ++(5:2) node[left] {\footnotesize $j_n(X)=U^{\dag}_tXU_t$};
    \foreach \x  in {6.75,7.25}
        \shadedraw [ball color=yellow] (\x,2.5,13) circle (0.25cm);
    \foreach \x  in {6.5,7,7.5}
        \shadedraw [ball color=yellow] (\x,2.5,13.5) circle (0.25cm);
    \foreach \x  in {6.25,6.75,7.25,7.75}
        \shadedraw [ball color=yellow] (\x,2.5,14) circle (0.25cm);
    \foreach \x  in {6.5,7,7.5}
        \shadedraw [ball color=yellow] (\x,2.5,14.5) circle (0.25cm);
    \foreach \x  in {6.75,7.25}
        \shadedraw [ball color=yellow] (\x,2.5,15) circle (0.25);

    \foreach \x  in {7} 
        \shadedraw [ball color=yellow] (\x,3,13.25) circle (0.25cm);
    \foreach \x  in {6.75,7.25}
        \shadedraw [ball color=yellow] (\x,3,13.75) circle (0.25cm);
    \foreach \x  in {6.5,7,7.5}
        \shadedraw [ball color=yellow] (\x,3,14.25) circle (0.25cm);
    \foreach \x  in {6.75,7.25}
        \shadedraw [ball color=yellow] (\x,3,14.75) circle (0.25cm);
    \foreach \x  in {7}
        \shadedraw [ball color=yellow] (\x,3,15.25) circle (0.25);
\end{scope}
\end{tikzpicture}
\caption{\label{fig:OpenSystem} Open system of a cluster of atoms interacting with an optical field. The input field described by quantum noises $A_t,A^{\dag}_t$ and the output operators after the interactions with the atoms expressed as an unitary evolution. X denotes observables of the system, atomic cluster in this case, couple to to the environment evolving by a unitary satisfying a QSDE.}
\end{figure}

Let us consider an example of a system that is driven by quantum noises and describe its dynamics in the language of quantum probability. Our physical system consists of two Jaynes-Cummins systems connected in series with the first cavity driven by a laser beam. The connections between the J-C units are quantized optical fields that may be thought of as quantum noises described above. We will derive a quantum master equation for this system,assumed to be in an initial X-state = $\begin{bmatrix} \rho_{11} & 0 & 0 & \rho_{14} \\ 0 & \rho_{22} & \rho_{23} & 0 \\ & \rho_{32} & \rho_{33} & 0 \\ \rho_{41} & 0 & 0 & \rho_{44} \end{bmatrix}$, and estimate the entanglement between the atoms in the two cavities and quantify it using the concurrence measure. A driven J-C system can be formulated in the language of well known input-output theory \cite {Gardiner1985} and the composite systems can be studied using its extension SLH framework. in this formalism, a quantum system interacting with an environment, quantum noises, can be characterized by three parameters, S - scattering matrix (discrete Poissonian noise) that connects input and output discrete channels like a generalized beam splitter.  The L operators are the Lindbladians that describe noise channels, H - the Hamiltonian describing the coherent evolution of the system.   
Let SLH for a single J-C system be described by the following parameters with two decoherence channels: 
\begin {align*}
S &= \mathbb{I}. &\text{ \color{red}  The identity matrix.} \\
L_1 &= \sqrt{\kappa}\times{a}. &\text{  \color{red} Channel for decay of the cavity through mirror} \\
 & &\text{ \color{red} where a is the annihilation operator}.\\
L_2 &= \sqrt{\gamma}\sigma. &\text{  \color{red} Atomic decay due to spontaneous} \\
& &\text{ \color{red} emission into outside modes.} \\
H     &= \Delta\times\sigma^{\dag}\otimes\sigma.   &\text{ \color{red} Detuning from atomic resonance} \\
       & + \Theta\times{a}^\dag\otimes{a}.  &\text{  \color{red} Detuning from cavity resonance} \\
       & + i\times{g}\times(\sigma\otimes{a}^\dag - \sigma^\dag\otimes{a}). &\text{  \color{red} Atom-mode coupling, i = sqrt(-1)}.
\end {align*}
We want to connect two such J-C systems in cascade by feeding one of the outputs of the first cavity to that of the second and drive the first one by a laser with strength $\alpha$ that has the SLH parameters $\left(S= \begin{pmatrix} 1 & 0 & 0 \\ 0 & 1 & 0 \\ 0 & 0 & 1 \end{pmatrix}, L = \begin{pmatrix} \alpha \\ 0 \\ 0\end{pmatrix}, H = 0\right)$ as depicted in \text {\color{blue} Figure \ref {fig:JC-C-O}}.

\begin{figure} \label {fig:JC-C-O}
\includegraphics[width=\columnwidth]{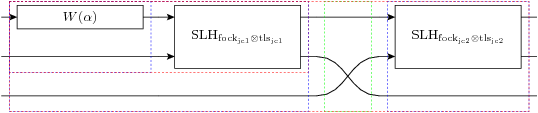}
\caption{\label{fig:JC-C-O} Cascade of two Jaynes-Cummins open systems driven by a laser input.}
\end{figure}
Then, the combined system can be described by the SLH parameters:

 \begin {equation}
 S = \begin{pmatrix} {\rm 1} & {\rm 0} & {\rm 0} \\ {\rm 0} & {\rm 0} & {\rm 1} \\ {\rm 0} & {\rm 1} & {\rm 0}\end{pmatrix}.
 \end {equation}
 \begin {equation}
 L = \begin{pmatrix}  \alpha +  \sqrt{\kappa} {a_{{{\rm fock}}_{{\rm jc1}}}} +  \sqrt{\kappa} {a_{{{\rm fock}}_{{\rm jc2}}}} \\  \sqrt{\gamma} {\sigma_{{\rm g},{\rm e}}^{{{\rm tls}}_{{\rm jc2}}}} \\  \sqrt{\gamma} {\sigma_{{\rm g},{\rm e}}^{{{\rm tls}}_{{\rm jc1}}}}\end{pmatrix}.  
\end {equation} 
 \begin {align*}
 H &= \Delta {\Pi_{{\rm e}}^{{{\rm tls}}_{{\rm jc1}}}} +  \Delta {\Pi_{{\rm e}}^{{{\rm tls}}_{{\rm jc2}}}} + i g \left( {a_{{{\rm fock}}_{{\rm jc1}}}^\dagger} {\sigma_{{\rm g},{\rm e}}^{{{\rm tls}}_{{\rm jc1}}}} -  {a_{{{\rm fock}}_{{\rm jc1}}}} {\sigma_{{\rm e},{\rm g}}^{{{\rm tls}}_{{\rm jc1}}}}\right) +  \\
 & i g \left( {a_{{{\rm fock}}_{{\rm jc2}}}^\dagger} {\sigma_{{\rm g},{\rm e}}^{{{\rm tls}}_{{\rm jc2}}}} -  {a_{{{\rm fock}}_{{\rm jc2}}}} {\sigma_{{\rm e},{\rm g}}^{{{\rm tls}}_{{\rm jc2}}}}\right) +  \\
 & \frac{i}{2} \left( - \alpha \sqrt{\kappa} {a_{{{\rm fock}}_{{\rm jc1}}}^\dagger} +  \sqrt{\kappa} \overline{\alpha} {a_{{{\rm fock}}_{{\rm jc1}}}}\right) +  \frac{i}{2} \left( - \sqrt{\kappa}  \left( \alpha +  \sqrt{\kappa} {a_{{{\rm fock}}_{{\rm jc1}}}}\right)  {a_{{{\rm fock}}_{{\rm jc2}}}^\dagger} +  \sqrt{\kappa}  \left( \overline{\alpha} +  \sqrt{\kappa} {a_{{{\rm fock}}_{{\rm jc1}}}^\dagger}\right)  {a_{{{\rm fock}}_{{\rm jc2}}}}\right) +  \\ 
 & \Theta  {a_{{{\rm fock}}_{{\rm jc1}}}^\dagger} {a_{{{\rm fock}}_{{\rm jc1}}}} + \Theta  {a_{{{\rm fock}}_{{\rm jc2}}}^\dagger} {a_{{{\rm fock}}_{{\rm jc2}}}} 
 \end {align*}
In the above, the scattering matrix S is no longer the identity owing to the channel crossing at the second stage of the circuit. This crossing is required as we want to feed only one output of the first cavity, the field that interacted with its atom, to the second J-C system. The three noise channels of the combined system is derived using circuit algebra that contains terms for decay through the mirrors combined and the two spontaneous atomic decays fed into the environment as shown in the figure ($\text{\color{blue}Figure \ref {fig:JC-C-O}}$). The coherent evolution of the composite system is described by a Hamiltonian that simply adds the atomic detuning resonance and cavity detuning resonance terms. The atom and cavity field mode coupling terms of individual cavities are added up in the composite system. The extra terms in the Hamiltonian account for the interaction of the input laser with the field mode of the first cavity and the output photons of the first J-C system interacting with the field modes of the second cavity. In principle the equation of motion of the composite system can be derived manually and the QNET \cite {Tezak2012} software automates this process. We will discuss another example of a composite system that involves a feedback loop where the software clearly provides an advantage in deriving the equations.

Let us simplify this system by making reasonable approximations based on adiabatic elimination \cite {Luc2008} where the bath excitations are removed in the limit of strong coupling between the system and the environment. 

\begin {defn} Adiabatic elimination: In the limit of strong coupling the excitations are removed from the system's description and the cavity field is directly coupled to the bath fields as can be seen from the modified channel vector below in equation \eqref{eqLsOp-1} and the Hamiltonian where the $\Theta$ terms representing the cavity detuning are missing. This approximation is applicable to a system strongly coupled to low temperature environment such as the electromagnetic field \cite {Gardiner1984}.
\end {defn}
\begin {equation} \label {eq:QSDEgen}
dU_t^k = \{({\color{red}S_{ij}^k} - \delta_{ij})d\Lambda_t^{ij} + {\color{blue}L_i^k}dA^{i\dag}_t  - {\color{red}S_{ij}^k}{\color{blue}L_i}^{k\dag}dA_t + K^kdt \}U_t.
\end {equation}
There are four assumptions on \text{\color{blue} Equation \eqref{eq:QSDEgen}}, which is a generalized form of \text{\color{blue} Equation \eqref{eq:QSDE}},  required for applying the adiabatic elimination model reduction method and we list them now.
\begin {defn} Assumption 1: $K^k + K^{k\dag} = L^{k\dag}L^k$; $S_{il}^k S_{jl}^{k\dag} = \delta_{ij}\mathbb{I}$; $S_{li}^{k\dag} S_{lj}^k = \delta_{ij}\mathbb{I}$;
\end {defn}
\begin {defn} Assumption 2, existence of operators independent of k: $K^k = k^2{Y} + kA + B$; $L_i^k = kF_i + G_i$; $S_{lj}^k = W_{ij}$;
\end {defn}
\begin {defn} Assumption 3, $P_0 = \mathbb{I} - P_1$, are projections onto ground and excited states: $P_1 Y_1^{-1} = Y_1^{-1}P_1$; $YY_1 P_1 Z P_0 = P_1 Z P_0$; $P_0 XP_1  Y_1^{-1}Y= P_0 XP_1$;
\end {defn}
\begin {defn} Assumption 4: $P_1 L_i = P_1 S_{ij} = 0$;
\end {defn}
For a more detailed treatment of this approximation and the assumptions involved we refer the reader to the work of Bouten et al. It is important to note that the adiabatic approximations are applied to individual cavities and the composite system is constructed out of the two limiting systems. The resulting quantum stochastic differential equation has the following SLH parameters:\cite {Luc2008}.
 \begin {equation}
  S = \begin{pmatrix} 1 & 0 & 0 \\ 0 & 0 & 1 \\ 0 & 1 & 0\end{pmatrix}.
  \end {equation}
 \begin {equation} \label {eqLsOp-1}
 L = \begin{pmatrix}  \alpha -  \frac{2 g}{\sqrt{\kappa}} {\sigma_{{\rm g},{\rm e}}^{{{\rm tls}}_{{\rm jc1}}}} +  \frac{2 g}{\sqrt{\kappa}} {\sigma_{{\rm g},{\rm e}}^{{{\rm tls}}_{{\rm jc2}}}} \\  \sqrt{\gamma} {\sigma_{{\rm g},{\rm e}}^{{{\rm tls}}_{{\rm jc2}}}} \\  \sqrt{\gamma} {\sigma_{{\rm g},{\rm e}}^{{{\rm tls}}_{{\rm jc1}}}}\end{pmatrix}.
 \end {equation}
 \begin {align*}
 H &= \Delta {\Pi_{{\rm e}}^{{{\rm tls}}_{{\rm jc1}}}} +  \frac{i \alpha}{\sqrt{\kappa}} g {\sigma_{{\rm e},{\rm g}}^{{{\rm tls}}_{{\rm jc1}}}} -  \frac{i g}{\sqrt{\kappa}} \overline{\alpha} {\sigma_{{\rm g},{\rm e}}^{{{\rm tls}}_{{\rm jc1}}}} + \\
 &  \Delta {\Pi_{{\rm e}}^{{{\rm tls}}_{{\rm jc2}}}} -  \frac{i \alpha}{\sqrt{\kappa}} g {\sigma_{{\rm e},{\rm g}}^{{{\rm tls}}_{{\rm jc2}}}} +  \frac{i g}{\sqrt{\kappa}} \overline{\alpha} {\sigma_{{\rm g},{\rm e}}^{{{\rm tls}}_{{\rm jc2}}}} -  \frac{2 i}{\kappa} g^{2}  {\sigma_{{\rm e},{\rm g}}^{{{\rm tls}}_{{\rm jc1}}}} {\sigma_{{\rm g},{\rm e}}^{{{\rm tls}}_{{\rm jc2}}}} +  \frac{2 i}{\kappa} g^{2}  {\sigma_{{\rm g},{\rm e}}^{{{\rm tls}}_{{\rm jc1}}}} {\sigma_{{\rm e},{\rm g}}^{{{\rm tls}}_{{\rm jc2}}}} 
 \end {align*}
Once we have the QSDE for the system we solve the equation numerically using QuTip package \cite {Nori2013} and here we provide the entanglement that is generated between the atoms of the two J-C systems, for the initial state of $\ket{ee}$, expressed in concurrence ($\text{\color{blue}Figure \ref {fig:JC-OP-ee}}$). In this configuration there is entanglement between the atomic degrees of freedom only when the atom in the first cavity is in excited state and not in ground state ($\text{\color{blue}Figure \ref {fig:JC-OP-ee}}$), the data for the other two initial states not shown.

\begin{figure}
\includegraphics[width=\columnwidth]{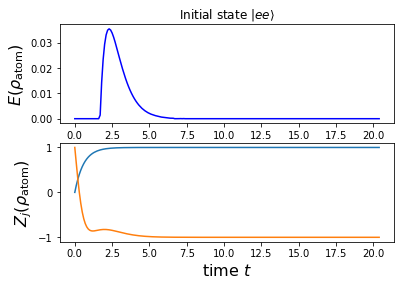}
\caption{\label{fig:JC-OP-ee} Entanglement between the atoms of the two cavities expressed in terms of concurrence for the initial state of $\ket{ee}$.}
\end{figure}
\bibliographystyle{abbrvnat}


\end{document}